\documentclass[lettersize,journal]{IEEEtran}
\usepackage{amsmath,amsfonts}
\usepackage{amssymb}
\usepackage{algorithm}
\usepackage[noend]{algpseudocode}
\usepackage{array}
\usepackage[caption=false,font=normalsize,labelfont=sf,textfont=sf]{subfig}
\usepackage{textcomp}
\usepackage{stfloats}
\usepackage{url}
\usepackage{verbatim}
\usepackage{graphicx}
\usepackage{cite}
\usepackage{enumitem}

\usepackage{bm}
\usepackage{mathtools}
\usepackage{commath}
\usepackage{yhmath}
\usepackage{mathrsfs}
\usepackage{amsthm}
\usepackage{hyperref}
\hypersetup{colorlinks=true, linkcolor=blue, breaklinks=true, urlcolor=blue, citecolor=blue}

\newcommand{\startpara}[1]{{%
\vskip5pt\noindent
{\bf #1.}}}

\newtheorem{defi}{\textbf{Definition}}
\newtheorem{thom}{\textbf{Theorem}}

\newtheorem{rek}{\textbf{Remark}}

\newtheorem{lema}{\textbf{Lemma}}

\newtheorem{cor}{\textbf{Corollary}}

\newcommand{\defiref}[1]{Definition \ref{#1}}
\newcommand{\thomref}[1]{Theorem~\ref{#1}}

\newcommand{\rekref}[1]{Remark~\ref{#1}}

\newcommand{\algoref}[1]{Algorithm \ref{#1}}

\newcommand{\lemaref}[1]{Lemma \ref{#1}}
\newcommand{\tabref}[1]{Table \ref{#1}}

\newcommand{\pnum}{N_p}
\newcommand{\enum}{N_e}
\newcommand{\pteam}{\mathscr{P}}
\newcommand{\eteam}{\mathscr{E}}

\newcommand{\goal}{\Omega_{\rm goal}}
\newcommand{\play}{\Omega_{\rm play}}
\newcommand{\gameregion}{\Omega}

\newcommand{\obstacleset}{\mathcal{O}}

\newcommand{\obstaclevertices}{V_{\textup{obs}}}
\newcommand{\obstacleverticesobs}{V_{\textup{obs}}^{\textup{gv}}}
\newcommand{\goalvertices}{V_{\textup{goal}}}

\newcommand{\ESP}{P_{\textup{ESP}}}
\newcommand{\ESPdist}{d_{\textup{ESP}}}
\newcommand{\wavefront}{W_{\textup{ESP}}}

\newcommand{\trianglepoints}{\mathcal{R}_{\textup{tri}}}
\newcommand{\sectorpoints}{\mathcal{R}_{\textup{sec}}}

\newcommand{\onsiteregion}{\mathcal{R}^{\textup{onsite}}}
\newcommand{\freespace}{\Omega_{\rm free}}
\newcommand{\ESPregion}{\mathcal{R}_{\textup{ESP}}}
\newcommand{\goalindex}{I_{\textup{goal}}}
\newcommand{\sector}{\mathcal{R}_{\textup{sec}}}

\begin{document}

\title{Pursuit Winning Strategies for Reach-Avoid Games with Polygonal Obstacles}

\author{Rui Yan, Shuai Mi,  Xiaoming Duan, Jintao Chen, and Xiangyang Ji
\thanks{The work of X. Duan was sponsored by Shanghai Pujiang Program under grant 22PJ1404900.}
\thanks{R. Yan is with the Department of Computer Science, University of Oxford, Oxford, OX1 3QD, UK. {\tt\small \{rui.yan@cs.ox.ac.uk\}}}
\thanks{S. Mi, J. Chen and X. Ji are with the Department of Automation, Tsinghua University, Beijing, 100084, China. {\tt\small \{mis15@tsinghua.org.cn, cjt16@tsinghua.org.cn, xyji@tsinghua.edu.cn\} }}
\thanks{X. Duan is with the Department of Automation, Shanghai Jiao Tong University, Shanghai, 200240, China. {\tt\small\{xduan@sjtu.edu.cn\} }}
}



\maketitle

\begin{abstract}
This paper studies a multiplayer reach-avoid differential game in the presence of general polygonal obstacles that block the players' motions. The pursuers cooperate to protect a convex region from the evaders who try to reach the region. We propose a multiplayer onsite and close-to-goal (MOCG) pursuit strategy that can tell and achieve an increasing lower bound on the number of guaranteed defeated evaders. This pursuit strategy fuses the subgame outcomes for multiple pursuers against one evader with hierarchical optimal task allocation in the receding-horizon manner. To determine the qualitative subgame outcomes that who is the game winner, we construct three pursuit winning regions and strategies under which the pursuers guarantee to win against the evader, regardless of the unknown evader strategy. First, we utilize the expanded Apollonius circles and propose the \emph{onsite pursuit winning} that achieves the capture in finite time. Second, we introduce convex goal-covering polygons (GCPs) and propose the \emph{close-to-goal pursuit winning} for the pursuers whose visibility region contains the whole protected region, and the goal-visible property will be preserved afterwards. Third, we employ Euclidean shortest paths (ESPs) and construct a pursuit winning region and strategy for the non-goal-visible pursuers, where the pursuers are firstly steered to positions with goal visibility along ESPs. In each horizon, the hierarchical optimal task allocation maximizes the number of defeated evaders and consists of four sequential matchings: capture, enhanced, non-dominated and closest matchings. Numerical examples are presented to illustrate the results.
\end{abstract}


\section{Introduction}
\emph{Problem motivation and description:} Differential games provide a proper mathematical framework to study the strategic behaviors of the players in the continuous state and action spaces \cite{RI:65}.
We consider a multiplayer reach-avoid differential game in a polygonal region with general polygonal obstacles. In this game, multiple pursuers protect a convex region against a number of malicious evaders. However, the obstacles block the motion of all players. We propose a hierarchical-matching receding-horizon cooperative pursuit strategy that is computationally efficient and can achieve a continuously improving lower bound on the number of evaders that can be defeated. Our study is motivated by the recent popularity of reach-avoid differential games \cite{RY-RD-XD-ZS-YZ:23,DS-VK:20,MC-ZZ-CJT:17,KM-JL:11,EG-DWS-AVM-MP:20,AVM-MP-DS-ZF:22,RY-XD-ZS-YZ-FB:22,LL-FD-ML-JC:21,PR-MK-YDM:23} and the lack of {literature} on efficient winning strategies for environments with obstacles. 


\emph{Literature review:} Reach-avoid differential games \cite{KM-JL:11,ZZ-RT-HH-CJT:12}, also known as two-target differential games \cite{WMG-MP:81,PC:96}, perimeter defense games \cite{DS-VK:20}, or target guarding problems \cite{HF-HHL:23,AVM-MP-DS-ZF:22}, are a class of differential games in which the evaders (or attackers) attempt to reach a target set while avoiding the capture by the pursuers (or defenders). The prevalence of such games in the last decade is a result of pressing needs to provide intelligent strategies for protecting critical infrastructures such as buildings and airports, against malicious air drones and robots. There are four common approaches to address reach-avoid differential games. The classical Hamilton-Jacobi (HJ) analysis, provides a general methodology, applies to nonlinear dynamics and is ideal for low-dimensional systems \cite{ZZ-RT-HH-CJT:12,IMM-AMB-CJT:05,IMM:02}. The characteristic method, relying on integrating backward from non-unique terminal sets, can generate closed-form solutions if the underlying singular surfaces are identified properly \cite{AVM-MP-DS-ZF:22,EG-DWC-MP:20}. Computational geometry methods, utilising geometric concepts such as Voronoi diagram \cite{ZZ-WZ-JD-HH-DMS-CJT:16}, Apollonius circle \cite{RY-ZS-YZ:19,JM-NK-BB:20,ZZ-JH-JX-YT:21}, function-based evasion space \cite{RY-XD-ZS-YZ-FB:22} and dominance region \cite{DWO-PTK-ARG:16}, have produced a diversity of strategies for simple-motion players. General forward reachable sets have also been applied to linear and nonlinear dynamics \cite{PR-MK-YDM:23,SYH-TS:17}. Learning-based approaches have emerged recently and been proved to be empirically superior, but still require more future effort for performance guarantees \cite{JS-KC-CAR-SG:23,ESL-LFZ-AR-VK:23,JS-EB:22}.

Most of the works on reach-avoid differential games assume obstacle-free environments, that is, the players can move freely in the environment \cite{RY-XD-ZS-YZ-FB:22,RY-RD-XD-ZS-YZ:23,RY-ZS-YZ:20-1,RY-ZS-YZ:19-2,RY-ZS-YZ:19,AVM-MP-DS-ZF:22,SV-SS-NS:22,DS-VK:20,JS-EB:19,DS-DM-MD:21,XC-JY:22}. However, the addition of obstacles enables the study of more complex and realistic scenarios \cite{NMTK-KGV:22}. For example, robots need to reach a pre-specified region in the obstacle-rich environments while avoiding other non-cooperative robots or moving objects, such as in urban search and rescue scenarios. Such games can also provide a framework to study the territory defense problems in urban and forest conflicts \cite{SB-SDB-AVM-ET-DWC:23,PC:96}. The air drones for package delivery must avoid the buildings, no-fly zones and other air drones, and reach the destinations safely.

Pursuit-evasion differential games in the presence of obstacles have been studied in the literature, where the capture is the unique competition goal. For example, Oyler \emph{et al.} \cite{DWO-PTK-ARG:16} provided the dominance region in the presence of a line segment or a triangle obstacle. Sometimes the obstacles generate both motion and visibility constraints while the pursuer's goal is to maintain the evader's visibility at all times \cite{EL-IB-UR-LB-RM:22,RZ-SB:18}. Zhou \emph{et al.} \cite{ZZ-JRS-DS-HH-CJT:20} proved that three pursuers are sufficient and sometimes necessary to guarantee the capture of an equal-speed evader in a bounded, two-dimensional arena with obstacles. The pursuer can learn and exploit the obstacle-avoidance mechanism of the evader to drive the latter to a trap position instead of capturing it \cite{YL-JH-CC-XG:22}. As a popular tool to address the safety considerations, control barrier functions are integrated to ensure the safety of the pursuers or evaders against obstacles \cite{NMTK-KGV:22}.


However, reach-avoid differential games with obstacles are more interesting but receive less attention due to the complexity. {HJ methods  \cite{MC-ZZ-CJT:17} intrinsically can deal with general obstacle environments, but they require the state space discretization (thus suffer from loss of accuracy), and are limited to systems up to about five states in practice \cite{MC-SLH-MSV-SB-CJT:18}.} The same scalability issue exists when approximating the game as a finite-state model and then synthesizing the strategies using LTL algorithms \cite{LS-TGT-AARN-AK-ACE:22}. Efficient algorithms have been proposed with restrictions, such as discrete-time dynamics and linear-quadratic approximations \cite{DRA-DPN-DFK-JFF:22}. Motion planning approaches can be combined to generate obstacle-free paths when the pursuers instead guard the target by herding the evaders into a safe area \cite{VSC-DP:21}. 

Current approaches for reach-avoid differential games in the presence of obstacles suffer from at least one of the following drawbacks: the curse of dimensionality due to the state space  discretization, heavy computational burden, model approximations and lack of winning guarantees (i.e., the construction of winning regions is missing). The winning regions \cite{RI:65}, which are subsets of the state space, are very appealing and useful. A player can win if the game starts from a state in its own winning region, despite the opponent's strategy. As far as we know, although the winning regions for reach-avoid differential games without obstacles, have been treated in the literature \cite{RY-ZS-YZ:19,RY-RD-XD-ZS-YZ:23,RY-ZS-YZ:19-2,RY-XD-ZS-YZ-FB:22,RY-ZS-YZ:20-1,MD-DM-DS-AVM:21,EG-DWS-AVM-MP:20,DS-VK:20}, an efficient approach for computing winning regions and strategies for environments with obstacles is missing, and it becomes even harder for multiplayer games.    

\emph{Contributions:} Utilising computational geometry methods and the hierarchical matching, we propose a multiplayer onsite and close-to-goal (MOCG) pursuit strategy for reach-avoid differential games in the presence of general polygonal obstacles. This strategy is computed efficiently, and can provide and achieve a lower bound on the number of defeated evaders that continuously improves over time. The main contributions are as follows:


\begin{enumerate}
    \item For the subgame with multiple pursuers and one evader, we propose three pairs of pursuit winning regions and the corresponding strategies that guarantee the winning of pursuers: \emph{onsite pursuit winning}, goal-visible and non-goal-visible cases for \emph{close-to-goal pursuit winning}. All results apply to the bounded polygonal environments with general polygonal obstacles.
    
    \item For the onsite pursuit winning, we propose a multiplayer pursuit strategy based on expanded Apollonius circles in \cite{MD-DM-DS-AVM:21}, and then construct the related winning region. In this winning, the pursuers can guarantee to capture the evader in an obstacle-free area and in finite time, targeting scenarios where players are close to each other with no obstacles nearby.

    \item For the close-to-goal pursuit winning, if the pursuers are goal-visible, i.e., all points in the goal region are visible to the pursuers, we introduce convex goal-covering polygons (GCPs). We then propose a GCP-based multiplayer pursuit strategy and construct the related winning region. This strategy ensures that pursuers are goal-visible at all times. In this winning, at most two pursuers are needed.
    
    \item For the close-to-goal pursuit winning, if a pursuer is not goal-visible, we combine GCPs with Euclidean shortest path and propose a two-stage pursuit strategy. Then, we construct the related winning region based on the convex optimization. In this winning, the game will evolve to one of last two cases eventually.

    \item Finally, we present a hierarchical task assignment scheme to merge all these subgame outcomes. A receding-horizon MOCG pursuit strategy is proposed such that the number of guaranteed defeated evaders is maximized at each step and increases with iterations. 
\end{enumerate}

\emph{Paper organization:} We introduce reach-avoid differential games with polygonal obstacles in Section \ref{sec:problem-statement}. We present the onsite pursuit winning, the goal-visible and non-goal-visible cases for close-to-goal pursuit winning in Sections \ref{sec:onsite-pursuit-winning}, \ref{sec:goal-visible} and \ref{sec:non-goal-visible}, respectively. In Section \ref{sec:multiplayer-strategy}, we propose the MOCG pursuit strategy. Numerical results are presented in Section \ref{sec:simulation} and we conclude the paper in Section \ref{sec:conclusion}.

\emph{Notation:} 
Let $\mathbb{R}$ be the set of reals. Let $\mathbb{R}^n$ be the set of $n$-dimensional real column vectors and $\enVert[0]{\cdot}_2$ be the Euclidean norm. 
Let $\bm{x}^\top$ denote the transpose of a vector $\bm{x} \in \mathbb{R}^n$. Denote the unit disk in $\mathbb{R}^n$ by $\mathbb{S}^{n}$, i.e., $\mathbb{S}^{n}=\{ \bm{u}\in\mathbb{R}^n\, |\, \|\bm{u}\|_2 \leq 1 \}$. For a finite set $S$, let $|S|$ be the cardinality of $S$. All angles in this paper are in the range $[0, 2\pi)$. For two angles $\theta_1 $ and $ \theta_2$, let $D[\theta_1, \theta_2]$ be the set of angles from $\theta_1$ to $\theta_2$ in a counterclockwise direction, and $| D[\theta_1, \theta_2]|$ be the corresponding angle difference from $\theta_1$ to $\theta_2$, also called angle span of $D[\theta_1, \theta_2]$. For two points $\bm{x}_1, \bm{x}_2 \in \mathbb{R}^2$, let $\overline{\bm{x}_1\bm{x}_2}$ be the line segment with endpoints $\bm{x}_1$ and  $\bm{x}_2$, and $\sigma(\bm{x}_1, \bm{x}_2)$ be the angle of the vector from $\bm{x}_1$ to $\bm{x}_2$. For three distinct points $\bm{x}_1, \bm{x}_2, \bm{x}_3 \in \mathbb{R}^2$, let $\trianglepoints (\bm{x}_1, \bm{x}_2, \bm{x}_3) \subset \mathbb{R}^2$ be the triangle region with vertices $\bm{x}_1, \bm{x}_2$ and $ \bm{x}_3$, and $\sectorpoints (\bm{x}_1, \bm{x}_2, \bm{x}_3) \subset \mathbb{R}^2 $ be the unbounded circular sector by counterclockwise sweeping the line emanating from $\bm{x}_1$ to $\bm{x}_2$ until the line from $\bm{x}_1$ to $\bm{x}_3$. {Further notations are provided in \tabref{tab:notation}, which will be explained in more detail later.}

\begin{table*}
    \centering 
    \captionsetup{labelformat=empty}
    \renewcommand{\arraystretch}{1.4}
    {
    \begin{tabular}{ p{.13\linewidth}  p{.32\linewidth} p{.13\linewidth}  p{.32\linewidth} } 
       \multicolumn{4}{ c }{Table 1. Notation Table} \\
        \hline
        Symbol & Description & Symbol & Description\\
        \hline
    $\pteam$ & $\pnum$ pursuers $\{P_1,\dots,P_{\pnum}\}$ & $\eteam$ & $\enum$ evaders $\{E_1,\dots,E_{\enum}\}$ \\ 
    \hline 
    $\obstacleset$ & family of obstacles $\{ O_1, \dots, O_k \}$ & $\gameregion, \goal$ & game region, convex goal region\\
    \hline 
     $\play$ & play region $\gameregion \setminus (\cup_{O \in \obstacleset} O \cup \goal)$ & $\freespace$ & $\play \cup \goal$ where players can move freely \\
    \hline
    $\alpha_{ij} = v_{P_i} / v_{E_j}$ & speed ratio between $P_i$ and $E_j$ & $r_i$ & capture radius of $P_i$ \\
    \hline
    $P_c$ & pursuit coalition $\{ P_i \in \pteam \mid i \in c \}$ &  $X_c = \{ \bm{x}_{P_i} \}_{i \in c}$ & positions of all pursuers in coalition $P_c$\\
    \hline
    $X_{ij} = (\bm{x}_{P_i}, \bm{x}_{E_j})$ & state of the subgame between $P_i$ and $E_j$ & $X_{cj} = (X_c, \bm{x}_{E_j})$ & state of the subgame between coalition $P_c$ and $E_j$ \\
    \hline
    $\mathbb{A}$ & Apollonius circle and its interior &  $\mathbb{A}_{\delta}$ & expanded Apollonius circle and its interior
    \\
    \hline
    $D_{\mathcal{R}} (\bm{x}) \in [0, 2\pi)$ & direction range of $\bm{x}$ in $\mathcal{R}$ & $\varrho (X_{cj})$ & safe distance of the subgame between $P_c$ and $E_j$
    \\
    \hline
    $\ESP(\bm{x}_1, \bm{x}_2)$ & Euclidean shortest path (ESP) between  $\bm{x}_1$ and $\bm{x}_2$ & $\ESPdist(\bm{x}_1, \bm{x}_2)$ & ESP distance between $\bm{x}_1$ and $\bm{x}_2$ \\
    \hline
    $\ESPregion(\bm{x}, \ell)$ & ESP reachable region from $\bm{x}$ within distance $\ell$  & $\wavefront(\bm{x}, \ell)$ & wavefront from $\bm{x}$ with distance $\ell$ \\
    \hline
    \end{tabular}
    }
    \label{tab:notation}
\end{table*}

\section{Problem Statement} \label{sec:problem-statement}

\subsection{Reach-avoid differential games with polygonal obstacles}

Consider a planar reach-avoid differential game in a polygonal region $\gameregion \subset \mathbb{R}^2$ with $\pnum$ pursuers, $\pteam=\{P_1,\dots,P_{\pnum}\}$, and $\enum$ evaders, $\eteam=\{E_1,\dots,E_{\enum}\}$, as depicted in Fig. \ref{fig:game-illustration}. Each player is assumed to be a mass point and moves with simple motion as Isaacs states \cite{RI:65}, i.e., it is holonomic. The dynamics of the players are described by
\begin{equation}
\begin{aligned}
 \dot{\bm{x}}_{P_i}(t) & = v_{P_i} \bm{u}_{P_i}(t), & \ \bm{x}_{P_i}(0) & = \bm{x}_{P_i}^0, & \quad P_i & \in \pteam \\
  \dot{\bm{x}}_{E_j}(t) & = v_{E_j} \bm{u}_{E_j}(t), & \ \bm{x}_{E_j}(0) & = \bm{x}_{E_j}^0, &\quad E_j & \in \eteam 
\end{aligned}
\end{equation}
where $\bm{x}_{P_i}  \in \mathbb{R}^2$ and $ \bm{x}_{E_j} \in \mathbb{R}^2 $ are the positions of $P_i$ and $E_j$, respectively, and $\bm{u}_{P_i}$ and $\bm{u}_{E_j}$ are their control inputs that belong to the admissible control set $\mathbb{U} = \{ \bm{u} : [0, \infty) \to \mathbb{S}^2 \mid \bm{u} \textup{ is piecewise smooth} \}$. For $P_i$ and $E_j$, their maximum speeds are $v_{P_i} > 0$ and $v_{E_j} > 0$, respectively, and their initial positions are $\bm{x}_{P_i}^0 \in \mathbb{R}^2$ and $\bm{x}_{E_j}^0 \in \mathbb{R}^2$, respectively.

A family of obstacles in $\gameregion$, denoted by $\obstacleset = \{ O_1, \dots, O_k \}$, block the motions of the players, where each obstacle $O_i \subset \gameregion$ $(1 \leq i \leq k)$ is a simple polygon (all boundaries are excluded) and the closures of any two obstacles are disjoint. 
A \emph{convex} polygon $\goal$ in $\gameregion$, which is disjoint from all obstacles and the pursuers and the evaders compete for, is called \emph{goal region} (the green polygon in Fig.~\ref{fig:game-illustration}). Formally, we have $\goal = \{ \bm{x} \in \mathbb{R}^2 \mid A_i \bm{x} + \bm{b}_i \ge 0, i \in \goalindex \}$, where $A_i \in \mathbb{R}^{2 \times 2}, \bm{b}_i \in \mathbb{R}^2$ and $\goalindex$ is an index set. The game region $\gameregion$ minus all obstacle polygons and the goal region is called \emph{play region} and denoted by $\play$, i.e., $\play = \gameregion \setminus (\cup_{O \in \obstacleset} O \cup \goal)$.
We denote by $\freespace \coloneqq \play \cup \goal$ the union of the play region and the goal region, in which all players can move freely.

{We denote by $\alpha_{ij} = v_{P_i} / v_{E_j}$ the speed ratio between $P_i$ and $E_j$ and consider faster pursuers, i.e., $\alpha_{ij} > 1$.} Let $r_i \ge 0$ be the capture radius of $P_i$. An evader is \emph{captured} by $P_i$ if $P_i$ is pursuing the evader, their Euclidean distance is less than or equal to $r_i$ and the line segment connecting them is obstacle-free (i.e., visual radius capture).

The evasion team $\eteam$ aims to send as many evaders initially in the play region $\play$ as possible into the goal region $\goal$ via the paths in $\freespace$, before being captured by the pursuit team $\pteam$ which conversely strives to guard $\goal$ against $\eteam$. 

In this problem, we consider the state feedback information structure \cite{TB-GJO:99}. The pursuit and evasion teams make decisions about their current control inputs with the information of all players' current positions. The maximum speeds of all players, the capture radii of all pursuers and the information about $\gameregion$, $\play$, $\goal$ and $\obstacleset$ are known by two teams.


\begin{figure}
    \centering
    \includegraphics[width=0.90\hsize,height=0.52\hsize]{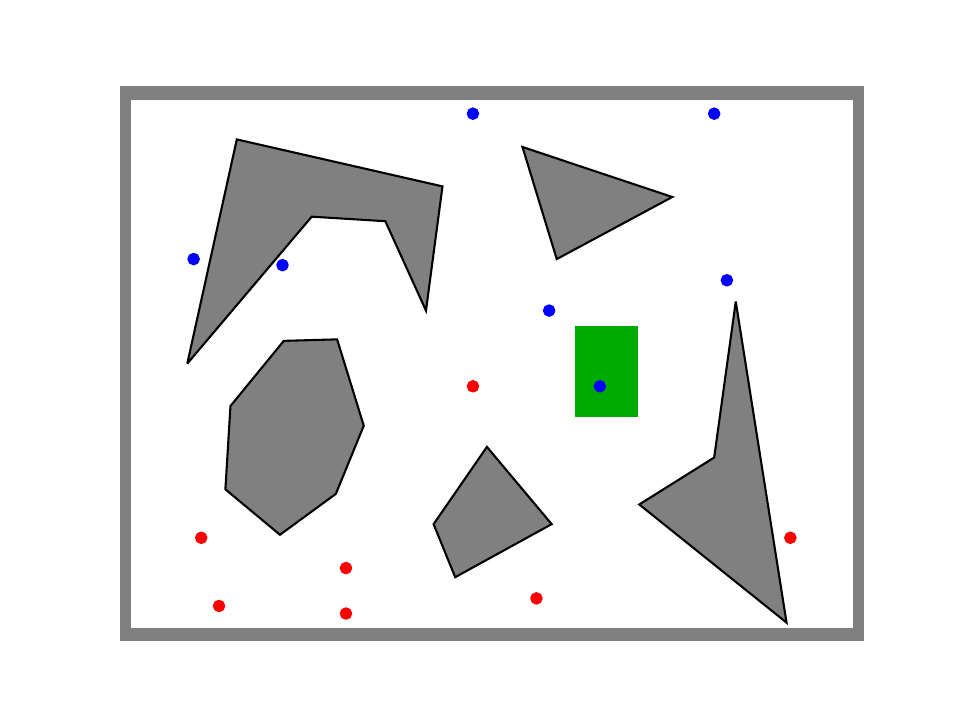}
    \put(-104,70){\footnotesize$P_i$}
    \put(-130,54){\footnotesize$E_j$}
    \put(-30,110){\footnotesize$\play$}
    \put(-67,62){\footnotesize$\goal$}
    \put(-215,115){\footnotesize$\gameregion$}
    \caption{Multiplayer reach-avoid differential games in the presence of black general polygonal obstacles, where the blue pursuers protect a green convex polygon $\goal$ against the red evaders who start from the play region $\play$ and aim to enter $\goal$. The obstacles block the motion of all players.}
    \label{fig:game-illustration}
\end{figure}
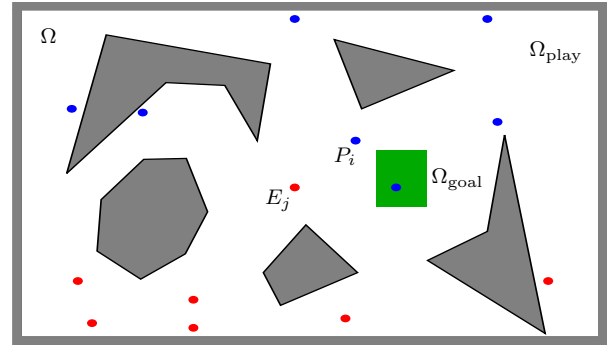

\subsection{Multiple-pursuer-one-evader subgames}
There exists complex, non-intuitive cooperation among team members and competition among players from different teams. Thus, it is challenging to synthesize strategies for multiple pursuers and multiple evaders directly. We split the whole game into many subgames, compute strategies for each subgame separately and generate team strategies by piecing together all subgame outcomes. This approach has been widely used to solve multiplayer games \cite{MC-ZZ-CJT:17,DS-VK:20,EG-DWS-AVM-MP:20,RY-XD-ZS-YZ-FB:22,RY-RD-XD-ZS-YZ:23}. {Notably, the very relevant work \cite{MC-ZZ-CJT:17} by Chen et al. only considered one-pursuer-one-evader subgames and solved multiplayer games by fusing pairwise outcomes via the maximum matching. In this paper, we consider the subgames with multiple pursuers and one evader and fuse the subgame outcomes via hierarchical optimal task allocation.} 

We introduce several definitions and notations for a set of pursuers. For any non-empty index set $c \in 2^{ \{1, \dots, \pnum \}}$, let $P_c = \{ P_i \in \pteam \mid i \in c \}$ be a \emph{pursuit coalition} containing pursuer $P_i$ if $i \in c$. We denote by $X_c$ and $U_c$ the positions and the control inputs of all pursuers in $P_c$, respectively, i.e., $X_c = \{ \bm{x}_{P_i} \}_{i \in c}$ and $U_c = \{\bm{u}_{P_i} \}_{i \in c}$. Let the product of $|c|$ sets $\mathbb{U}$, i.e., $\mathbb{U}_c = \mathbb{U} \times \cdots \times \mathbb{U}$, be the admissible control set for $P_c$. Let $X_{cj}=(X_c,\bm{x}_{E_j})$ be the state (i.e., position) of the subgame between $P_c$ and $E_j$, and $X_{ij}=(\bm{x}_{P_i},\bm{x}_{E_j})$ be the subgame state (i.e., position) if $c=\{i\}$.  
Any variable that depends on the players' positions $\bm{x}_{P_i}(t)$ or $\bm{x}_{E_j}(t)$ is time-dependent. For notational convenience, we omit $t$ unless needed for clarity.

For a subgame between a pursuit coalition $P_c$ and an evader $E_j$, the pursuit winning is defined as follows.

\begin{defi}[Pursuit winning]
A pursuit coalition $P_c  \in 2^{\pteam}$ guarantees to win against an evader $E_j \in \eteam$ at a state $X_{cj}$, if there exists a pursuit strategy $U_c\in \mathbb{U}_c$ such that starting from $X_{cj}$, one of the conditions holds regardless of $E_j$'s strategy:
\begin{enumerate}
    \item  $P_c$ can guarantee to capture $E_j$ in $\play$; 
    \item  $P_c$ can guarantee to delay $E_j$ from entering $\goal$ forever unless $E_j$ is captured.
\end{enumerate}
The state $X_{cj}$ and the strategy $U_c$ are called a pursuit winning state and a corresponding winning strategy, respectively. A set of pursuit winning states are called a pursuit winning region.
\end{defi}

{
\subsection{Problems of interest}
 In this paper, our goal is to propose strategies for the pursuit team $\pteam$ that can guarantee to win against as many evaders simultaneously as possible (i.e., capture evaders or delay them from entering $\goal$ forever). The idea is to divide the pursuers into many disjoint coalitions and then assign each coalition with a carefully designed pursuit strategy to a particular evader. 

To this end, we first construct three pursuit winning regions (i.e., positions of a pursuit coalition and an evader) from which the pursuit coalition can win regardless of the evader's strategy. These three pursuit winning regions correspond to different scenarios informally defined as follows and detailed later.

\begin{defi}
    [Three pursuit winning regions] For a pursuit coalition and an evader, an \emph{onsite} pursuit winning region corresponds to the scenario where the pursuers are close to the evader with no obstacles nearby. The goal-visible case for the \emph{close-to-goal} pursuit winning region corresponds to the scenario where the pursuers can visibly see the whole goal region, while the non-goal-visible case corresponds to the scenario where the visibility is not satisfied.
\end{defi}

Then, three pursuit winning regions are used during the play to determine whether a pursuit coalition can win against an evader. By collecting all such subgame outcomes, we generate disjoint pursuit coalitions and assign them to evaders in real time such that the number of defeated evaders is maximized, where each pursuit coalition will either take the strategy that ensures its winning or some heuristic strategy detailed below. The assignment will change if a new assignment guaranteeing to defeat more evaders is generated as the game evolves.
}

\section{Onsite pursuit winning} \label{sec:onsite-pursuit-winning}

For a subgame between a pursuit coalition $P_c$ and an evader $E_j$, this section focuses on the onsite pursuit winning where $E_j$ can be captured before the obstacles are used to assist its evasion. We present the pursuit winning region and the related winning strategy for $P_c$.  

\subsection{Expanded Apollonius-circle based pursuit strategy}

We first recall an expanded Apollonius circle based pursuit strategy in the \emph{obstacle-free} plane for one pursuer against one evader due to  Dorothy et al. \cite{MD-DM-DS-AVM:21}. For a state $X_{ij}$, we define two variables:
\begin{equation}\label{eq:RA-xA}
  R_A = \frac{\alpha_{ij} \| \bm{x}_{P_i} - \bm{x}_{E_j} \|_2 }{\alpha_{ij}^2 - 1}, \, \bm{x}_A = \frac{ \alpha_{ij}^2 \bm{x}_{E_j} - \bm{x}_{P_i} }{ \alpha_{ij}^2 - 1} \,.  
\end{equation}
Then, the Apollonius circle \cite{RI:65} and its interior, denoted by $\mathbb{A}$, and its expanded version with a small constant $\delta>0$, denoted by $\mathbb{A}_{\delta}$ (with the green boundary in Fig.~\ref{fig:onsite-pursuit-winnning}(a)), are respectively computed as follows:
\begin{equation} \label{eq:apollonius-circle}
\begin{aligned}
    \mathbb{A} & = \{ \bm{x} \in \mathbb{R}^2 \mid \| \bm{x}- \bm{x}_A \|_2 \leq R_A \}, \\
    \mathbb{A}_{\delta} & = \{ \bm{x} \in \mathbb{R}^2 \mid\| \bm{x}- \bm{x}_A \|_2 \leq R_A + \delta  \}.
\end{aligned}
\end{equation}
It is well-known \cite{RI:65} that in obstacle-free cases, starting from the state $X_{ij}$, $E_j$ can reach any point in $\mathbb{A}$ no later than $P_i$, while $P_i$ can reach any point in $\mathbb{R}^2 \setminus \mathbb{A}$ before $E_j$. Recently, by expanding $\mathbb{A}$ to the larger region $\mathbb{A}_{\delta}$, the work \cite{MD-DM-DS-AVM:21} proposed the following pursuit strategy such that $P_i$ guarantees to capture $E_j$ in $\mathbb{A}_{\delta}$ and in a finite time, regardless of $E_j$'s strategy.

\begin{lema}[Pursuit strategy, \cite{MD-DM-DS-AVM:21}]\label{lema:onsite-pursuit-strategy-1v1}
In the absence of obstacles, from time $t$, if $P_i$ adopts the pursuit strategy $\bm{u}_{P_i}(\tau) = \frac{\bm{z}(\tau)}{ \| \bm{z}(\tau) \|_2 }$
for all $\tau \ge t$, where $\bm{z}(\tau)$ is given by
\begin{equation*}
    \alpha_{ij} (R_A(t) + \delta - R_A(\tau)) \frac{\bm{x}_{E_j}(\tau) - \bm{x}_{P_i}(\tau)}{ \| \bm{x}_{E_j}(\tau) - \bm{x}_{P_i}(\tau) \|_2 } + \bm{x}_A(\tau) - \bm{x}_A(t)
\end{equation*}
then regardless of $E_j$'s strategy, $P_i$ can ensure
\begin{enumerate}[label=(\roman*)]
    \item $\mathbb{A}(\tau) \subset \mathbb{A}_{\delta}(t)$ for all $ \tau \ge t$;

    \item $E_j$ is captured in $\mathbb{A}_{\delta}(t)$ under $r_i=0$ in a finite time.
\end{enumerate}
\end{lema}

{
We provide motivations for adopting the union set $\mathbb{A}_{\delta}$ of the expanded Apollonius circle and its interior, and some intuition on how it works.
}

{
\begin{rek}
    It is known \cite{RI:65,PW-MP-KP:19} that the set of points that $P_i$ and $E_j$ can reach along their minimum distances with a time difference $r_i/v_{P_i} \ge 0$ is a Cartesian oval, which degenerates into an Apollonius circle for $r_i=0$. We employ the expanded Apollonius circle instead of the Cartesian oval in the pursuit strategy for both $r_i=0$ and $r_i>0$, for two reasons: 1) it is unclear whether a similar result to \lemaref{lema:onsite-pursuit-strategy-1v1} still holds for a pursuit strategy based on the Cartesian oval with $r_i >0$; 2) if $P_i$ can guarantee to capture $E_j$ under $r_i=0$, then $P_i$ can also ensure the capture under $r_i > 0$. The extension of the strategy in \lemaref{lema:onsite-pursuit-strategy-1v1} to the Cartesian oval is not straightforward, and we leave it for future work.
\end{rek}

\begin{rek}
    Since $P_i$ has no access to the current control input of $E_j$, $P_i$ and $E_j$ may move towards different points at the boundary of $\mathbb{A}$ at the state $X_{ij}$. Thus, $P_i$ cannot guarantee to capture $E_j$ in $\mathbb{A}$, provided that $r_i=0$. However, introducing $\mathbb{A}_{\delta}$ which strictly expands $\mathbb{A}$, can generate a pursuit strategy for $P_i$ that ensures the capture in $\mathbb{A}_{\delta}$. The general idea behind this strategy is that instead of focusing on points in $\mathbb{A}$, $P_i$ can think a bit further (i.e., measured by $\delta$) and consider the points that $P_i$ can reach strictly before $E_j$. For the interested readers, please refer to \cite{MD-DM-DS-AVM:21} which proposed this strategy, for details.
\end{rek}
}

\subsection{Onsite pursuit winning regions and strategies}

We extend the results in \lemaref{lema:onsite-pursuit-strategy-1v1} to the case of a pursuit coalition $P_c$ against an evader $E_j$ in the presence of obstacles. 

For a pursuer $P_i$ ($i \in c$) against $E_j$, we construct a region as follows. For a state $X_{ij}$, let $\bm{x}_{T_1}$ and $\bm{x}_{T_2}$ be two distinct points in $\mathbb{A}_{\delta}$ such that the line segments $\overline{\bm{x}_{P_i} \bm{x}_{T_1}}$ and $\overline{\bm{x}_{P_i} \bm{x}_{T_2}}$ are tangent to the circle $\mathbb{A}_{\delta}$, see Fig.~\ref{fig:onsite-pursuit-winnning}(a). By \eqref{eq:apollonius-circle}, we have
\begin{align*}
    \| \bm{x}_{T_k} - \bm{x}_A \|_2 = R_A + \delta,\, ( \bm{x}_{T_k} - \bm{x}_{P_i} )^{\top} (\bm{x}_{T_k} - \bm{x}_A) = 0
\end{align*}
for $k = 1,2$, from which we obtain
\begin{equation*}
\begin{aligned}
    \bm{x}_{T_1} & = \bm{x}_A + \frac{ \ell_1 (\bm{x}_{P_i} - \bm{x}_A) + \ell_2 (\bm{x}_{P_i}- \bm{x}_A)^\circ }{ \| \bm{x}_{P_i} - \bm{x}_A \|_2^2 }, \\
    \bm{x}_{T_2} & = \bm{x}_A + \frac{ \ell_1 (\bm{x}_{P_i}- \bm{x}_A) - \ell_2 (\bm{x}_{P_i}- \bm{x}_A)^\circ }{ \| \bm{x}_{P_i} - \bm{x}_A \|_2^2 } ,
\end{aligned}
\end{equation*}
where for a vector $\bm{x} = [x, y]^\top \in \mathbb{R}^2$, let $\bm{x}^\circ = [y, - x]^\top$, and
\begin{equation*}
    \ell_1 = ( R_A + \delta )^2, \ \ell_2 = ( R_A + \delta ) \sqrt{\| \bm{x}_{P_i} - \bm{x}_A \|_2^2 - \ell_1 } .
\end{equation*}
Then, we construct a region induced by $X_{ij}$ as follows:
\begin{align}\label{eq:onsite-region}
    \onsiteregion_{ij}  = \trianglepoints(\bm{x}_{P_i}, \bm{x}_{T_1}, \bm{x}_{T_2}) \cup \mathbb{A}_{\delta},
\end{align}
which is the union of two regions with black and green boundaries in Fig.~\ref{fig:onsite-pursuit-winnning}(a). We use $\onsiteregion_{ij}$ to construct the following pursuit winning region for $P_i$ against $E_j$.

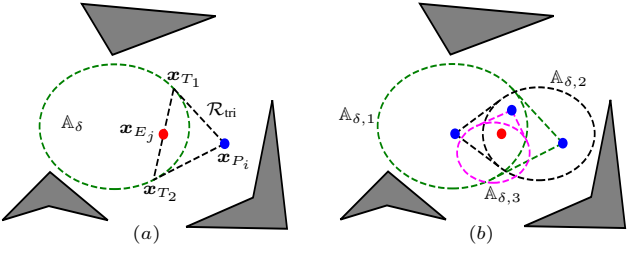
\begin{figure}
    \centering
    \includegraphics[width=0.44\hsize,height=0.35\hsize]{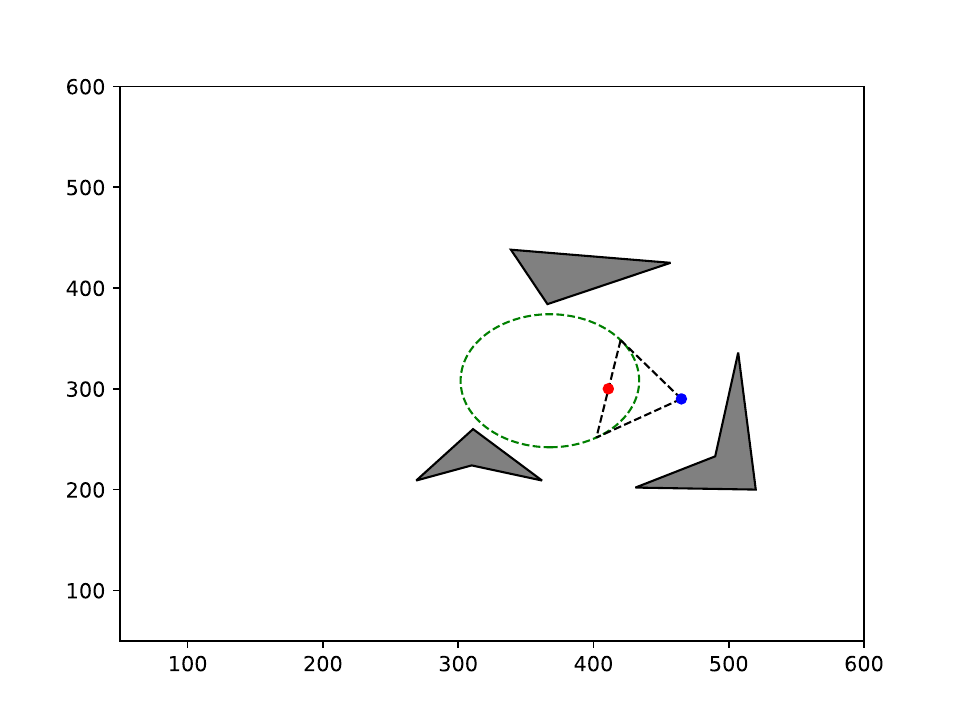} \quad \
    \includegraphics[width=0.44\hsize,height=0.35\hsize]{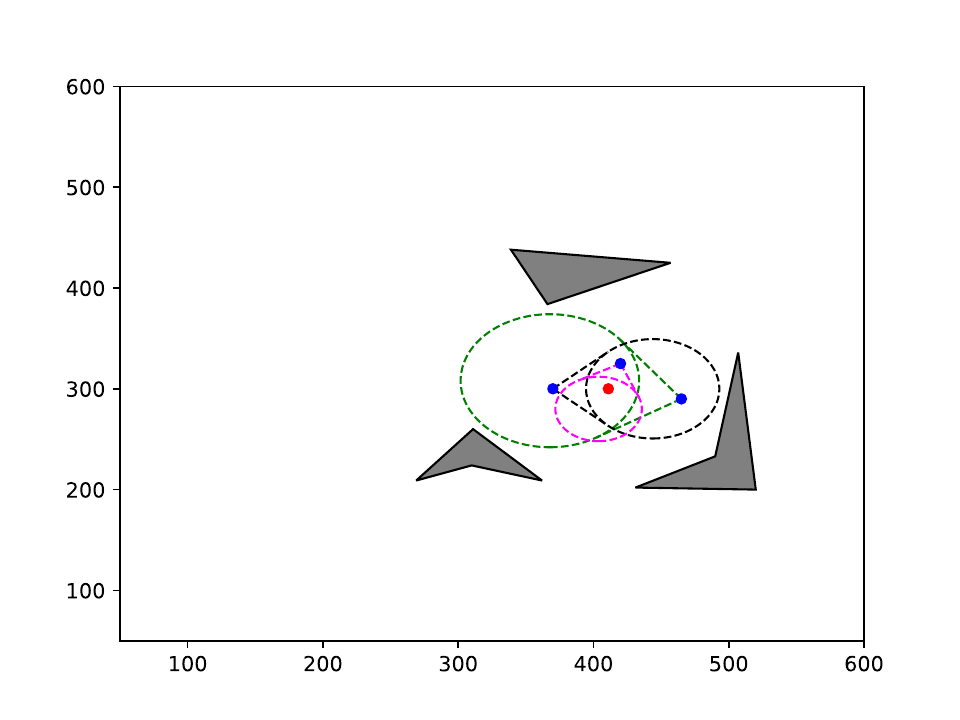}
    \put(-215,40){\scriptsize$\mathbb{A}_{\delta}$}
    \put(-156,27){\scriptsize$\bm{x}_{P_i}$}
    \put(-193,38){\scriptsize$\bm{x}_{E_j}$}
    \put(-175,58){\scriptsize$\bm{x}_{T_1}$}
    \put(-184,15){\scriptsize$\bm{x}_{T_2}$}
    \put(-110, 43){\scriptsize$\mathbb{A}_{\delta,1}$}
    \put(-30, 57){\scriptsize$\mathbb{A}_{\delta,2}$}
    \put(-55, 13){\scriptsize$\mathbb{A}_{\delta,3}$}
    \put(-160,45){\scriptsize$\trianglepoints$}
    \put(-61,-2){\scriptsize$(b)$}
    \put(-188,-2){\scriptsize$(a)$}
    \caption{Onsite pursuit winning. $(a)$ one pursuer: if two regions $\mathbb{A}_{\delta}$ and $\trianglepoints$ surrounded by the green expanded Apollonius circle and black triangle, respectively, are obstacle-free, then $P_i$ guarantees to capture $E_j$ in $\mathbb{A}_{\delta}$ in a finite time. $(b)$ pursuit coalition: if every pursuer in a coalition can capture an evader individually, then the evader will be captured in the intersection of all regions bounded by expanded Apollonius circles, i.e., $\cap_{i=1}^3 \mathbb{A}_{\delta,i}$.}
    \label{fig:onsite-pursuit-winnning}
\end{figure}

\begin{thom}[Onsite pursuit winning]\label{thom:onsite-pursuit-winning}
At time $t$, if the positions $X_{ij}$ of $P_i$ and $E_j$ satisfy 
\begin{equation}\label{eq:onsite-winning-condition-1v1}
    \onsiteregion_{ij}(t) \cap O = \emptyset \textup{ for all }  O \in  \obstacleset \textup{ and }  \mathbb{A}_{\delta}(t) \cap \goal = \emptyset
\end{equation}
then using the pursuit strategy in \lemaref{lema:onsite-pursuit-strategy-1v1} for all $\tau \ge t$, $P_i$ can guarantee to capture $E_j$ in $\mathbb{A}_{\delta}(t)$  despite $E_j$'s strategy. 
\end{thom}
\begin{proof}
We first consider the case when there is no obstacle. If $P_i$ adopts the pursuit strategy in \lemaref{lema:onsite-pursuit-strategy-1v1}, then  starting from $X_{ij}$ at time $t$, $P_i$ can guarantee to capture $E_j$ in $\mathbb{A}_{\delta}(t)$ under the point capture $r_i = 0$  and in a finite time. This also applies to our case $r_i \ge 0$ due to a larger capture range. When the obstacles are present, the same conclusion follows if we can prove that $P_i$ using this pursuit strategy never leaves the region $\onsiteregion_{ij}(t)$ before $E_j$ is captured despite $E_j$'s strategy. This is because $\onsiteregion_{ij}(t)$ is obstacle-free and $\mathbb{A}_{\delta}(t)$ does not intersect with $\goal$ by \eqref{eq:onsite-winning-condition-1v1}.

Let $t_{\textup{capture}} > t$ be the first time when $E_j$ is captured by $P_i$ if there is no obstacle. Next, we prove that $\bm{x}_{P_i}(\tau) \in \onsiteregion_{ij}(t)$ for all $t \leq \tau \leq t_{\textup{capture}}$. By \lemaref{lema:onsite-pursuit-strategy-1v1}(i), we have $\mathbb{A}(\tau) \subset \mathbb{A}_{\delta}(t)$ and thus by \eqref{eq:onsite-region}, $\mathbb{A}(\tau) \subset \onsiteregion_{ij}(t)$. Under the pursuit strategy in \lemaref{lema:onsite-pursuit-strategy-1v1}, $P_i$ always moves along the direction $\bm{z}(\tau)$. Then the conclusion follows if we can show that for any $t \leq \tau \leq t_{\textup{capture}}$, $\bm{z}(\tau)$ consistently points from $\bm{x}_{P_i}(\tau)$ to $\mathbb{A}(\tau)$. That is, the straight line emanating from $\bm{x}_{P_i}(\tau)$ along the direction $\bm{z}(\tau)$ intersects with $\mathbb{A}(\tau)$. 

Since $\bm{x}_A(\tau) - \bm{x}_{P_i}(\tau) = \alpha_{ij}^2(\bm{x}_{E_j}(\tau) - \bm{x}_{P_i}(\tau)) / (\alpha_{ij}^2 - 1)$ using \eqref{eq:RA-xA}, then $\bm{z}(\tau)$ has the equivalent form
\begin{equation*}
    \alpha_{ij} (R_A(t) + \delta - R_A(\tau)) \frac{\bm{x}_A(\tau) - \bm{x}_{P_i}(\tau)}{ \| \bm{x}_A(\tau) - \bm{x}_{P_i}(\tau) \|_2 } + \bm{x}_A(\tau) - \bm{x}_A(t).
\end{equation*}
Without changing $\bm{u}_{P_i}(\tau)$, by rearranging, the direction $\bm{z}(\tau)$ can be further rewritten as follows:
\begin{equation*}
\begin{aligned}
     & \bm{x}_A(\tau) - \bm{x}_{P_i}(\tau) + \frac{ \| \bm{x}_A(\tau) - \bm{x}_{P_i}(\tau) \|_2 (\bm{x}_A(\tau) - \bm{x}_A(t) ) }{\alpha_{ij} (R_A(t) + \delta - R_A(\tau))} \\
     & \coloneqq \bm{x}_{P_i}^{\star}(\tau) -  \bm{x}_{P_i}(\tau),
\end{aligned}
\end{equation*}
where $\bm{x}_{P_i}^{\star}(\tau)$ is the sum of the first and third terms. Therefore, $\bm{u}_{P_i}(\tau)$ can be represented as 
\begin{equation}
    \bm{u}_{P_i}(\tau) = \frac{\bm{x}_{P_i}^{\star}(\tau) -  \bm{x}_{P_i}(\tau)}{ \| \bm{x}_{P_i}^{\star}(\tau) -  \bm{x}_{P_i}(\tau) \|_2 },
\end{equation}
that is, $P_i$ always moves towards the point $\bm{x}_{P_i}^{\star}(\tau)$. Since two disks $\mathbb{A}(\tau) $ and $\mathbb{A}_{\delta}(t)$ satisfy $\mathbb{A}(\tau) \subset \mathbb{A}_{\delta}(t)$, then
\[
    \| \bm{x}_A(\tau) - \bm{x}_A(t) \|_2 \leq R_A(t) + \delta - R_A(\tau).
\]
Therefore, we have 
\begin{equation*}
\begin{aligned}
    \| \bm{x}_{P_i}^{\star}(\tau) - \bm{x}_A(\tau) \|_2 & = \frac{ \| \bm{x}_A(\tau) - \bm{x}_{P_i}(\tau) \|_2 \| \bm{x}_A(\tau) - \bm{x}_A(t) \|_2 }{\alpha_{ij} (R_A(t) + \delta - R_A(\tau))} \\
    & \leq \frac{\| \bm{x}_A(\tau) - \bm{x}_{P_i}(\tau) \|_2 }{ \alpha_{ij} } \\
    & = \frac{ \alpha_{ij} \| \bm{x}_{P_i}(\tau) - \bm{x}_{E_j}(\tau) \|_2 }{ \alpha_{ij}^2 - 1 } = R_A(\tau),
\end{aligned}
\end{equation*}
where \eqref{eq:RA-xA} is used in the last two equalities. Thus, $\bm{x}_{P_i}^{\star}(\tau) \in \mathbb{A}(\tau)$, implying that $P_i$ always moves towards a point in $\mathbb{A}(\tau)$, i.e., $\bm{x}_{P_i}^{\star}(\tau)$. Since $\bm{x}_{P_i}(t) \in \onsiteregion_{ij}(t)$ and $\mathbb{A}(\tau) \subset \onsiteregion_{ij}(t)$, we have $\bm{x}_{P_i}(\tau) \in \onsiteregion_{ij}(t)$ for all $t \leq \tau \leq t_{\textup{capture}}$. Since \eqref{eq:onsite-winning-condition-1v1} says there is no obstacle in $\onsiteregion_{ij}(t)$, then $P_i$ can capture $E_j$ without hitting any obstacle.
\end{proof}

Using \thomref{thom:onsite-pursuit-winning}, we then have the following onsite pursuit winning region and strategy for $P_c$ against $E_j$.

\begin{lema}[Onsite pursuit winning for coalitions]\label{lema:onsite-winning-nv1}
    At time $t$, if the positions $X_{cj}$ of $P_c$ and $E_j$ are such that \eqref{eq:onsite-winning-condition-1v1} holds for all $i \in c$, then using the pursuit strategy in \lemaref{lema:onsite-pursuit-strategy-1v1} for each pursuer $P_i$ in $P_c$ and all $\tau \ge t$, $P_c$ can guarantee to capture $E_j$ in $\cap_{i \in c} \mathbb{A}_{\delta,i}(t)$, where $\mathbb{A}_{\delta,i}(t)$ is the expanded Apollonius circle and its interior for $P_i$ against $E_j$ at time $t$.
\end{lema}
\begin{proof}
    For each pursuer $P_i$ ($i \in c$), since \eqref{eq:onsite-winning-condition-1v1} holds, then by \thomref{thom:onsite-pursuit-winning}, $P_i$ can capture $E_j$ in $\mathbb{A}_{\delta,i}(t)$ on its own. Thus, this further implies that 
    $P_c$ can capture $E_j$ in $\cap_{i \in c} \mathbb{A}_{\delta,i}(t)$.
\end{proof}

\begin{rek} \label{rek:onsite-winning}
    By \lemaref{lema:onsite-winning-nv1}, if more pursuers satisfy \eqref{eq:onsite-winning-condition-1v1} against one evader, then the evader is guaranteed to be captured in a smaller { or equal} region and in a shorter { or equal} time. For instance, in Fig.~\ref{fig:onsite-pursuit-winnning}(b), three pursuers guarantee to capture the evader in the intersection of the regions bounded by three expanded Apollonius circles. Since \eqref{eq:onsite-winning-condition-1v1} and the pursuit strategy in \lemaref{lema:onsite-pursuit-strategy-1v1} are both independent of the other pursuers, the resulting cooperative pursuit strategy inherently allows for the distributed implementation. 
\end{rek}

\section{Close-to-goal pursuit winning I: \\ Goal-visible case} \label{sec:goal-visible}

For the case when the condition \eqref{eq:onsite-winning-condition-1v1} required by the onsite pursuit winning does not hold, we propose another two pursuit winning regions and strategies for the subgames. This section presents the first case where the pursuers can visually see the whole goal region.

\subsection{Goal-visible points and convex goal-covering polygons}

In order to construct the pursuit winning regions and strategies, we introduce several new geometric concepts and propose methods to verify or construct them. Utilising these concepts, we classify the points in $\freespace$ into two categories. Then, we investigate the pursuit winning regions and strategies when the pursuers are in each category separately, in the current and the next sections, respectively. Recall that $\freespace = \play \cup \goal$. The classification is as follows.

\begin{defi}[Goal-visible point]\label{defi:visibly-goal-covering-point}
 A point $\bm{x} \in \freespace$ is called  \textup{goal-visible} if for any point $\bm{y} \in \goal$, $\bm{y}$ is visible to $\bm{x}$, i.e., the line segment $\overline{\bm{x} \bm{y}}$ does not intersect with obstacles in $\obstacleset$.
\end{defi}

\defiref{defi:visibly-goal-covering-point} shows that if a pursuer is at a goal-visible point (i.e., $\bm{x}$ in Fig.~\ref{fig:goal-visible-point}(a)), then it can reach any point in $\goal$ via an obstacle-free line segment. However, a pursuer at a non-goal-visible point (i.e., $\bm{x}$ in Fig.~\ref{fig:goal-visible-point}(b)) needs to change its direction in order to reach some point in $\goal$. If a goal-visible point is in $\play$, we next define a pair of critical points for verifying the goal-visible property and computing the pursuit strategies.

\begin{figure}[!h]
    \centering
    \includegraphics[width=0.44\hsize]{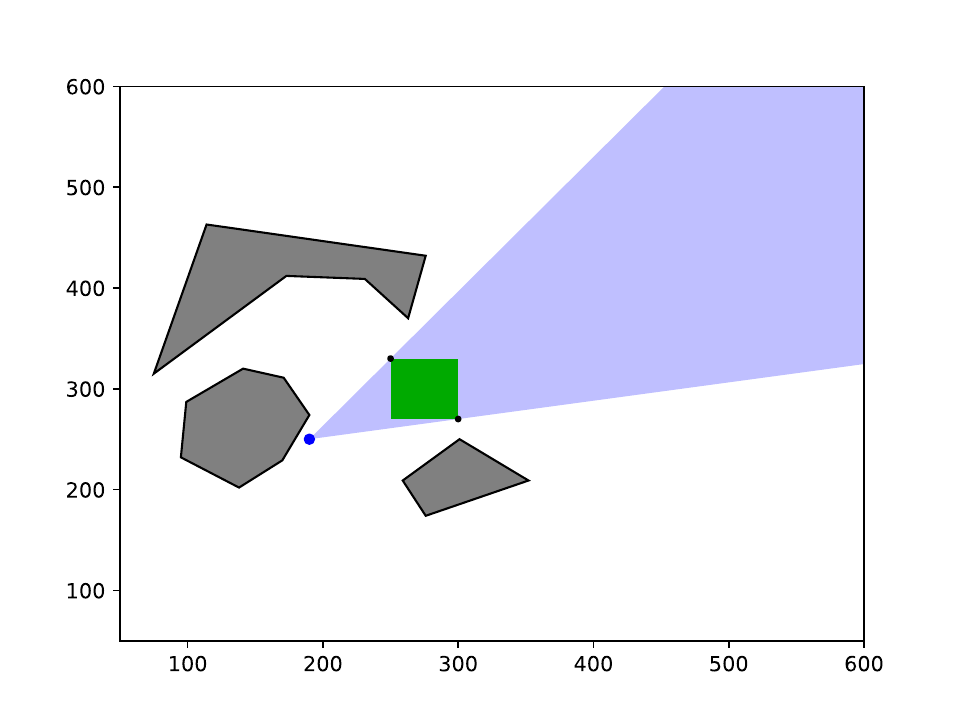} \quad \
    \includegraphics[width=0.44\hsize]{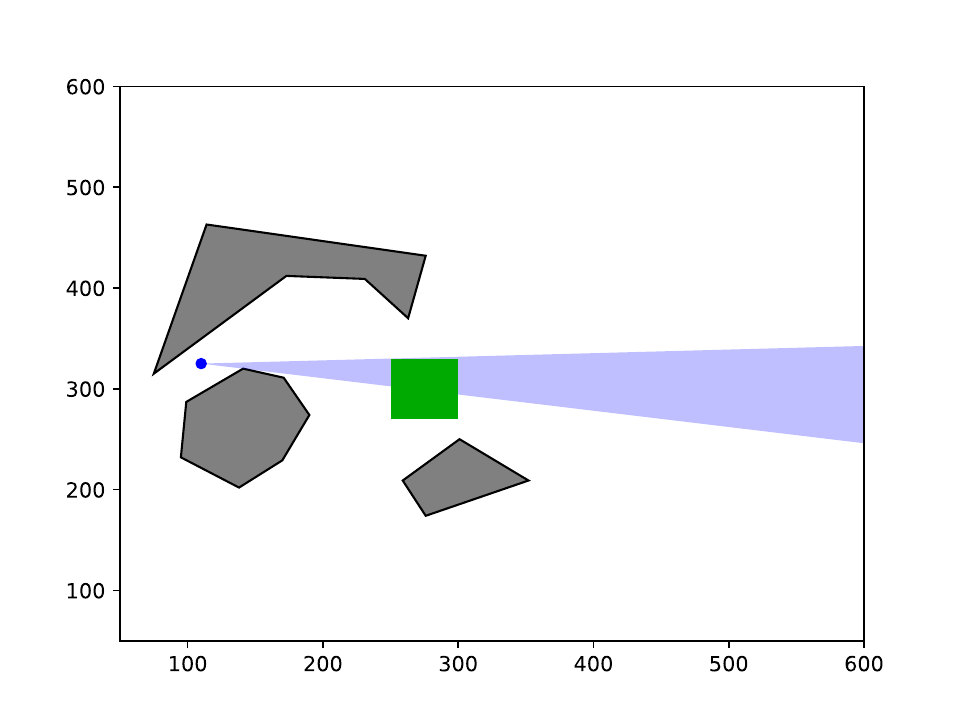}
    \put(-192,17){\scriptsize$\bm{x}$}
    \put(-177,50){\scriptsize$\bm{x}_2$}
    \put(-148,25){\scriptsize$\bm{x}_1$}
    \put(-102, 40){\scriptsize$\bm{x}$}
    \put(-153,78){\scriptsize$\sectorpoints$}
    \put(-18,51){\scriptsize$\sectorpoints$}
    \put(-61,0){\scriptsize$(b)$}
    \put(-188,0){\scriptsize$(a)$}
    \caption{$(a)$ goal-visible point $\bm{x}$ with a pair $(\bm{x}_1, \bm{x}_2)$ of minimum-covering points: all points in $\goal$ are visible to $\bm{x}$ and $\goal \subset \sectorpoints(\bm{x}, \bm{x}_1, \bm{x}_2)$. $(b)$ non-goal-visible point $\bm{x}$: at least one point in $\goal$ is not visible to $\bm{x}$.}
    \label{fig:goal-visible-point}
\end{figure}

\begin{defi}[Pair of minimum-covering  points]\label{defi:minimum-covering-points}
For a goal-visible point $\bm{x} $ in $\play$, $(\bm{x}_1, \bm{x}_2) \in \goal^2$ is called a \emph{pair of minimum-covering  points} for $\bm{x}$ if 
    $\goal \subset \sectorpoints (\bm{x}, \bm{x}_1, \bm{x}_2)$. 
\end{defi}

Let $\goalvertices$ be the set of vertices of $\goal$. As $\goal$ is a convex polygon, a pair of minimum-covering points for a goal-visible point $\bm{x}$ in $\play$ can be computed by verifying the condition $\goal \subset \sectorpoints (\bm{x}, \bm{x}_1, \bm{x}_2)$ for any two vertices $\bm{x}_1, \bm{x}_2 \in \goalvertices$. For instance, the blue region in Fig.~\ref{fig:goal-visible-point}(a) is the circular sector $\sectorpoints (\bm{x}, \bm{x}_1, \bm{x}_2)$ that covers $\goal$.

Now we are ready to show the verification approach.

\begin{lema}[Verifying a goal-visible point]\label{lema:verifying-gvp}
    A point $\bm{x} \in \freespace$ is goal-visible if and only if one of the two conditions holds: $i)$ $\bm{x} \in \goal$; $ii)$ $\bm{x} \in \play$ and $\trianglepoints (\bm{x}, \bm{x}_{1}, \bm{x}_{2}) \cap O = \emptyset $ for all $O \in \obstacleset$, where $(\bm{x}_1, \bm{x}_2) \in \goal^2$ is a pair of minimum-covering points for $\bm{x}$.
\end{lema}

\begin{proof}
    Regarding $i)$, all points in $\goal$ are goal-visible by \defiref{defi:visibly-goal-covering-point}, as $\goal$ is obstacle-free and convex. 

    {
    Regarding $ii)$, if $\bm{x} \in \play$ and  $\trianglepoints (\bm{x}, \bm{x}_{1}, \bm{x}_{2}) \cap O = \emptyset $ for all $O \in \obstacleset$, then by \defiref{defi:minimum-covering-points} and the convexity of $\goal$, the line segment $\overline{\bm{x} \bm{y}}$ does not intersect with the obstacles for all $\bm{y} \in \goal$, i.e., $\bm{x}$ is goal-visible, where $\bm{x}_1, \bm{x}_2 \in \mathbb{R}^2$ are a pair of minimum-covering points for $\bm{x}$. Conversely, if $\bm{x} \in \play$ is goal-visible, then by  \defiref{defi:visibly-goal-covering-point}, the line segment $\overline{\bm{x} \bm{y}}$ does not intersect with the obstacles for all $\bm{y} \in \goal$ and thus for all $\bm{y}$ at the line segment $\overline{\bm{x}_1 \bm{x}_2}$. Therefore,  $\trianglepoints (\bm{x}, \bm{x}_{1}, \bm{x}_{2}) \cap O = \emptyset $ for all $O \in \obstacleset$, which completes the proof.
    }
\end{proof}

We construct a class of polygons, which can ensure the consistent goal-visible property and will be used in synthesizing pursuit winning strategies below.

\begin{defi}[Convex goal-covering polygon]\label{defi:convex-goal-covering-polygon}
    A polygon $\mathcal{R} \subset \mathbb{R}^2$ is a \emph{convex goal-covering polygon (GCP)} for a goal-visible point $\bm{x} \in \freespace$  if $i)$ $\mathcal{R}$ is convex, $ii)$ $\bm{x}  \in \mathcal{R}$,
    $iii)$ $\mathcal{R} \subset \freespace$ 
    and $iv)$ $\goal$ is covered by $\mathcal{R}$ (i.e., $\goal \subset \mathcal{R}$).
\end{defi}

As an illustration, the orange regions in Fig.~\ref{fig:convex-GCP} are convex GCPs for the underlying points $\bm{x}$, differing in whether $\bm{x}$ is at the boundary of the  convex GCP. The following lemma shows that such convex GCPs always exist and possess the following consistent property.

\begin{figure}[!h]
    \centering
    \includegraphics[width=0.44\hsize]{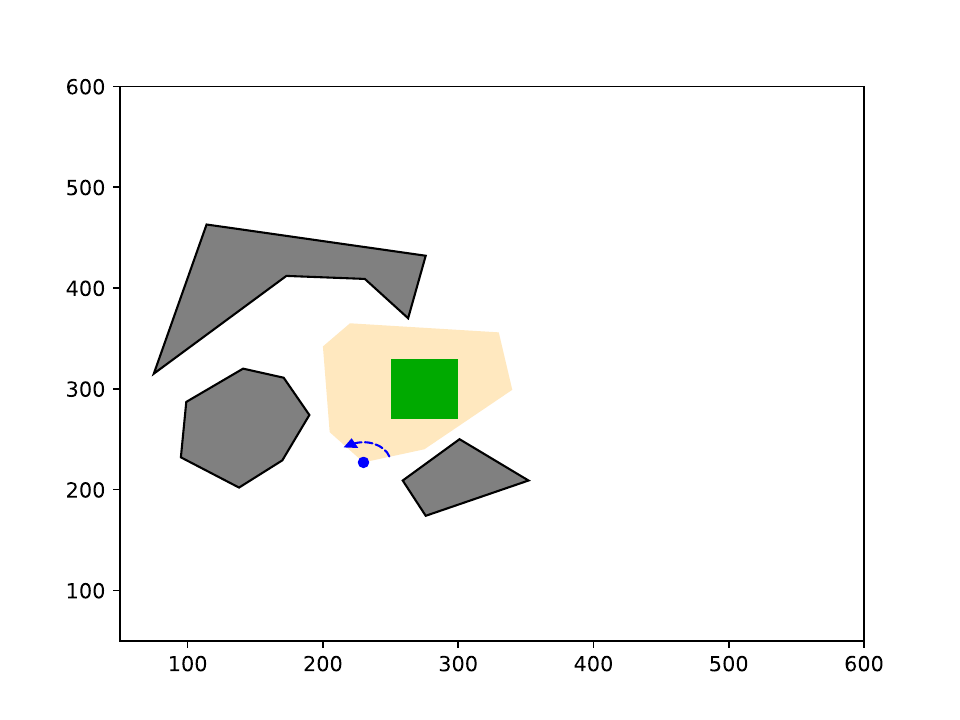} \quad \
    \includegraphics[width=0.44\hsize]{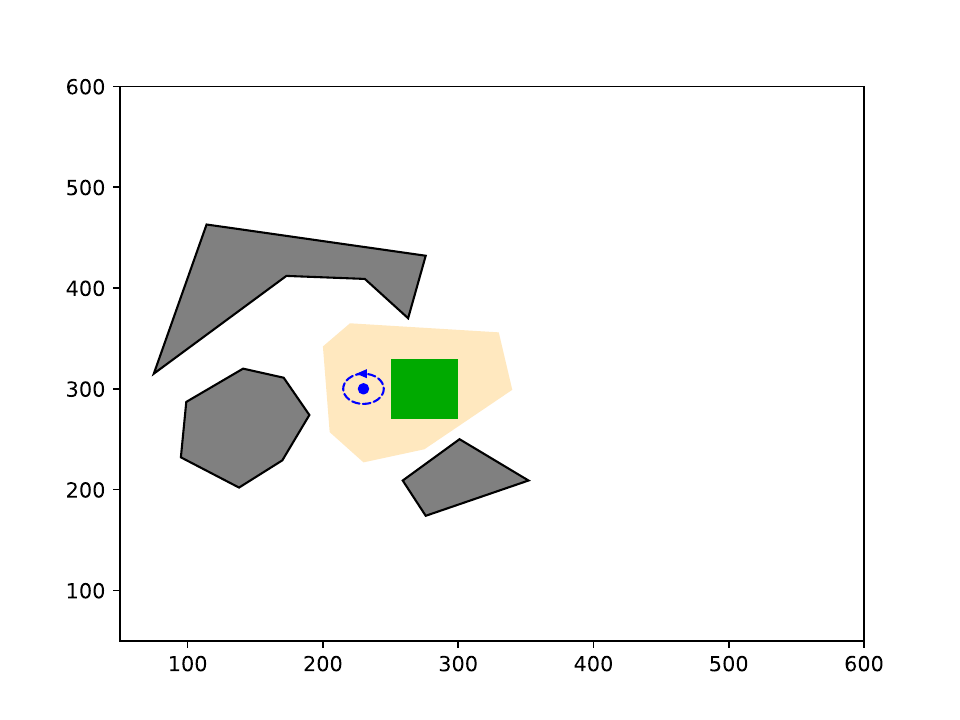}
    \put(-180,10){\scriptsize$\bm{x}$}
    \put(-51, 28){\scriptsize$\bm{x}$}
    \put(-153,57){\scriptsize$\mathcal{R}$}
    \put(-24,57){\scriptsize$\mathcal{R}$}
    \put(-61,0){\scriptsize$(b)$}
    \put(-188,0){\scriptsize$(a)$}
    \caption{Convex goal-covering polygon (GCP). The orange polygon $\mathcal{R}$ is a convex GCP for $\bm{x}$, as $i)$ $\mathcal{R}$ is convex, $ii)$ $\bm{x}  \in \mathcal{R}$, $iii)$ $\mathcal{R} \subset \freespace$ 
    and $\goal \subset \mathcal{R}$. $(a)$ if $\bm{x}$ is at the boundary of $\mathcal{R}$, the direction range $D_{\mathcal{R}}(\bm{x})$ is the blue direction range; otherwise $(b)$ $D_{\mathcal{R}}(\bm{x}) = [0, 2\pi)$.}
    \label{fig:convex-GCP}
\end{figure}
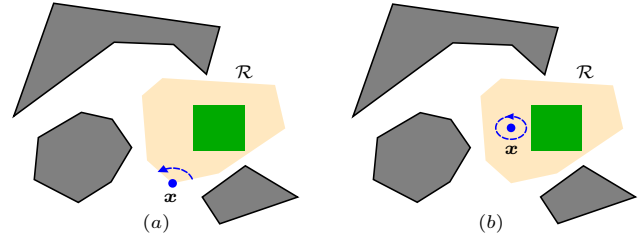

\begin{lema}[Existence and consistency]\label{lema:existence-consistency}
    If a point $\bm{x} \in \freespace$ is goal-visible, then there must exist a convex GCP for $\bm{x}$. If $\mathcal{R} \subset \freespace$ is a convex GCP for some goal-visible point, then all points in $\mathcal{R}$ are goal-visible.
\end{lema}
\begin{proof}
    Regarding the existence, consider a goal-visible point $\bm{x} \in \freespace$. If $\bm{x} \in \goal$, then $\goal$ is a natural convex GCP for $\bm{x}$. If $\bm{x} \in \play$, then by \defiref{defi:minimum-covering-points} there exist a pair of minimum-covering points $(\bm{x}_1, \bm{x}_2) \in \goal^2$ for $\bm{x}$, where $\bm{x}_1 $ and $ \bm{x}_2$ are two vertices of $\goal$, i.e., $\bm{x}_1, \bm{x}_2 \in \goalvertices$. Let $S$ be the set of vertices in $\goalvertices$ which are not in $\trianglepoints(\bm{x}, \bm{x}_1, \bm{x}_2)$. Since $\goal$ is a convex polygon and $\goal \subset  \sectorpoints (\bm{x}, \bm{x}_1, \bm{x}_2)$, then the vertex set $S \cup \{\bm{x}, \bm{x}_1, \bm{x}_2 \}$ forms a convex polygon $\mathcal{R}$ that covers $\goal$. Thus, $\mathcal{R}$ satisfies the conditions $i)$, $ii)$ and $iv)$ in \defiref{defi:convex-goal-covering-polygon}. As $\bm{x}$ is goal-visible, then $\trianglepoints (\bm{x}, \bm{x}_1, \bm{x}_2)$ is obstacle-free. Recall that $\goal \subset \freespace$. Hence, we obtain that $\mathcal{R} \subset \freespace$ and therefore $\mathcal{R}$ is a convex GCP for $\bm{x}$.

    Regarding the consistency, since $\mathcal{R}$ is convex and obstacle-free and covers $\goal$, then all points in $\mathcal{R}$ have obstacle-free line segments to any point in $\goal$, i.e., all points in $\mathcal{R}$ are goal-visible.
\end{proof}

\begin{rek}\label{rek:consistent-goal-visible}
\lemaref{lema:existence-consistency} implies that, if a pursuer is at a goal-visible point, then moving towards any point in a convex GCP (it indeed exists) for this goal-visible point can ensure that the pursuer is goal-visible consistently. 
\end{rek}

The following concept is proposed to characterize the set of moving directions in a convex GCP in \rekref{rek:consistent-goal-visible} that preserve the goal-visible property.  

\begin{defi}[Direction range in a convex GCP]\label{defi:direction-range}
    Let $\mathcal{R}$ be a convex GCP for a goal-visible point $\bm{x} \in \freespace$. The \emph{direction range} of $\bm{x}$ in $\mathcal{R}$, denoted by $D_{\mathcal{R}} (\bm{x}) \subset [0, 2\pi)$, is the set of angles $\{ \sigma(\bm{x}, \bm{y}) \mid \bm{y} \neq \bm{x}, \bm{y} \in \mathcal{R} \}$.
\end{defi}

The direction range can be easily identified as follows.

\begin{rek}\label{rek:identify-direction-range}
Let $\mathcal{R}$ be a convex GCP for a goal-visible point $\bm{x}$. { As in  Fig.~\ref{fig:convex-GCP}(a), if $\bm{x} $ is at the boundary of $ \mathcal{R}$, then $D_{\mathcal{R}}(\bm{x}) = D[\theta^L, \theta^U]$, where $\theta^L, \theta^U \in [0, 2 \pi)$ are two directions of the edges from $\bm{x}$ along $ \mathcal{R}$'s boundary. As in  Fig.~\ref{fig:convex-GCP}(b), if $\bm{x} $ is in the interior of $ \mathcal{R}$, then $D_{\mathcal{R}}(\bm{x}) = [0, 2 \pi)$.}
\end{rek}

\subsection{Goal-visible pursuit winning}
We consider the first category where all pursuers in a pursuit coalition $P_c$ are goal-visible (i.e., lie at goal-visible points). The following concept for the obstacle-free case introduced in \cite{RI:65} is required for the strategy synthesis below.

\begin{defi}[Evasion region]
Given a state $X_{cj}$, the \emph{evasion region} is the set of points in $\mathbb{R}^2$ that $E_j$ can reach prior to the capture by $P_c$, regardless of $P_c$'s strategy, in the absence of obstacles.
\end{defi}

According to \cite{RY-XD-ZS-YZ-FB:22}, the closure of the evasion region, denoted by $\mathbb{E}(c,j)$, is bounded, strictly convex and given by
\begin{equation}
    \mathbb{E}(c,j) = \{ \bm{x} \in \mathbb{R}^2 \mid f_{ij} (\bm{x}, X_{ij}) \ge 0, i \in c \},
\end{equation}
where $f_{ij} : \mathbb{R}^2 \times \mathbb{R}^2 \times \mathbb{R}^2 \to \mathbb{R}$ for $i \in c$ is defined as
\begin{equation*}
    f_{ij}(\bm{x}, X_{ij}) = \| \bm{x} - \bm{x}_{P_i} \|_2 - \alpha_{ij} \| \bm{x} - \bm{x}_{E_j} \|_2 - r_i .
\end{equation*}
Note that the Apollonius circle is a special case of the evasion region's boundary when $r_i=0$. Then, the \emph{safe distance} of $X_{cj}$, denoted by $\varrho(X_{cj})$, is defined as the signed distance between $\mathbb{E}(c,j)$ and $\goal$, which can be computed using the convex optimization below. 


\begin{defi}[Safe distance]\label{defi:obstacle-free-safe-distance}
For a state $X_{cj}$, let $(\bm{x}_I, \bm{x}_G)$ be the solution of the convex optimization problem $\mathcal{P}(X_{cj})$:
\begin{equation}\label{eq:convex-pbm-IG}
\begin{aligned}
    & \underset{(\bm{x},\bm{y})\in\mathbb{R}^2\times\mathbb{R}^2}{\textup{minimize}}
	&& \| \bm{x} - \bm{y} \|_2 \\
	&\textup{ subject to}&& f_{ij}(\bm{x},X_{ij})\ge 0,\; \forall i\in c\\
    & & & A_m \bm{y} + \bm{b}_m \ge 0, \;  \forall m \in \goalindex  \,.
\end{aligned}
\end{equation}
If 1) $\bm{x}_I = \bm{x}_G$ and $f_{ij}(\bm{x}_I,X_{ij}) > 0$ for all $i \in c$, or 2) $\bm{x}_I = \bm{x}_G$ and $A_m \bm{x}_I + \bm{b}_m > 0$ for all $m \in \goalindex$, then $\varrho (X_{cj}) = - \infty$, and $\varrho (X_{cj}) = \| \bm{x}_I - \bm{x}_G \|_2$ otherwise.
\end{defi}

In order to construct the pursuit winning region and strategy, we next present a sufficient condition on the pursuit strategies to ensure the goal-visible consistency.

\begin{lema}[Consistent goal-visible pursuit strategy]\label{lema:consistent-goal-visible-p-strategy}
    At time $t$, suppose that $\bm{x}_{P_i}$ is goal-visible. If a pursuit strategy $\bm{u}_{P_i}(\tau) = [\cos \theta_i(\tau), \sin \theta_i(\tau)]^\top$ is such that $P_i$ moves towards a point in a convex GCP $\mathcal{R}_i(\tau)$ for $\bm{x}_{P_i}(\tau)$, i.e., $\theta_i(\tau) \in D_{\mathcal{R}_i(\tau)}(\bm{x}_{P_i}(\tau))$ for all $\tau \ge t$, then $\bm{x}_{P_i}(\tau)$ is goal-visible for all $\tau \ge t$.
\end{lema}
\begin{proof}
    Since $\bm{x}_{P_i}(t)$ is goal-visible, then according to \lemaref{lema:existence-consistency} there exists a convex GCP $\mathcal{R}_i(t)$ for  $\bm{x}_{P_i}(t)$. Since $\bm{x}_{P_i}(t) \in \mathcal{R}_i(t)$ by \defiref{defi:convex-goal-covering-polygon} and $P_i$ moves towards a point in $\mathcal{R}_i(t)$ under the input $\bm{u}_{P_i}(t)$ (i.e., $\theta_i(t) \in D_{\mathcal{R}_i(t)}(\bm{x}_{P_i}(t)$), then after moving along $\bm{u}_{P_i}(t)$ with a small time interval $\Delta t > 0$, we have $\bm{x}_{P_i}(t + \Delta t) \in \mathcal{R}_i(t)$ (note that $\mathcal{R}_i(t)$ is convex). Then, the consistency in \lemaref{lema:existence-consistency} shows that $\bm{x}_{P_i}(t + \Delta t)$ is goal-visible for which there exists a convex GCP $\mathcal{R}_i(t + \Delta t)$. Thus, the conclusion follows by the similar argument for $\bm{x}_{P_i}(t + \Delta t)$ and the increasing iteration.
\end{proof}

Utilising the safe distance, goal visibility and convex GCPs, we next present the pursuit winning region and strategy for the goal-visible case of the close-to-goal pursuit winning, as illustrated in Fig.~\ref{fig:goal-visible-pursuit-winning}.

\begin{thom}[Goal-visible pursuit winning]\label{thom:goal-visible-winning}
    At time $t$, if the positions $X_{cj}$ of $P_c$ and $E_j$ are such that 
    \begin{enumerate}
        \item $\bm{x}_{P_i}$ is goal-visible for all $i \in c$; \label{itm:goal-visible-pwin-1}

        \item the safe distance is non-negative, i.e., $\varrho (X_{cj}) \ge 0$; \label{itm:goal-visible-pwin-2}
    \end{enumerate}
    then  $\bm{x}_{P_i}(\tau)$ is goal-visible for all $\tau \ge t$ and all $i \in c$, and $P_c$ can guarantee to win against $E_j$, regardless of $E_j$'s strategy, by using the pursuit strategy for all $i \in c$ as follows: 
    \begin{equation}\label{eq:goal-visible-pursuit-strategy}
       \bm{u}_{P_i}(\tau) =  [\cos \theta_i(\tau), \sin \theta_i(\tau)]^\top
    \end{equation}
    for all $\tau \ge t$, where
    \begin{equation}\label{eq:goal-visible-strategy-theta}
        \theta_i(\tau) = \textup{argmin}_{\theta \in D_{\mathcal{R}_i(\tau)}(\bm{x}_{P_i}(\tau))} |\theta - \sigma(\bm{x}_{P_i}(\tau), \bm{x}_I(\tau))|
    \end{equation}
    In \eqref{eq:goal-visible-strategy-theta},  $D_{\mathcal{R}_i(\tau)}(\bm{x}_{P_i}(\tau))$ is the direction range of $\bm{x}_{P_i} (\tau)$ in a convex GCP $\mathcal{R}_i(\tau)$ and $(\bm{x}_I(\tau), \bm{x}_G(\tau))$ is the optimal solution to $\mathcal{P}(X_{cj}(\tau))$ in \eqref{eq:convex-pbm-IG}. 
\end{thom}

\begin{figure}
    \centering
    \quad \includegraphics[width=0.85\hsize,height=0.35\hsize]{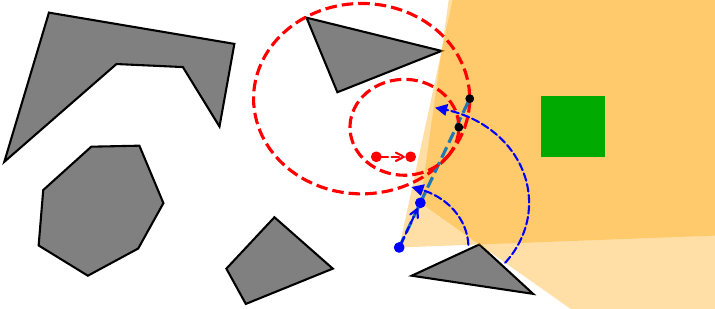}
    \put(-235,75){\scriptsize$(a)$}
    \put(-139,77){\scriptsize$\mathbb{E}$}
    \put(-116,54){\scriptsize$\mathbb{E}'$}
    \put(-118,41){\scriptsize$\bm{x}_{E_j}$}
    \put(-101,51){\scriptsize$\bm{x}_{E_j}'$}
    \put(-72,63){\scriptsize$\bm{x}_I$}
    \put(-76,45){\scriptsize$\bm{x}_I'$}
    \put(-107,13){\scriptsize$\bm{x}_{P_i}$}
    \put(-106,27){\scriptsize$\bm{x}_{P_i}'$}
    \put(-77,29){\scriptsize$\phi$}
    \put(-55,35){\scriptsize$\phi'$}
    \put(-20,65){\scriptsize$\mathcal{R}_i$}
    \put(-34,7){\scriptsize$\mathcal{R}_i'$}
    \put(-50,-7){}
        
    \quad \includegraphics[width=0.85\hsize,height=0.35\hsize] {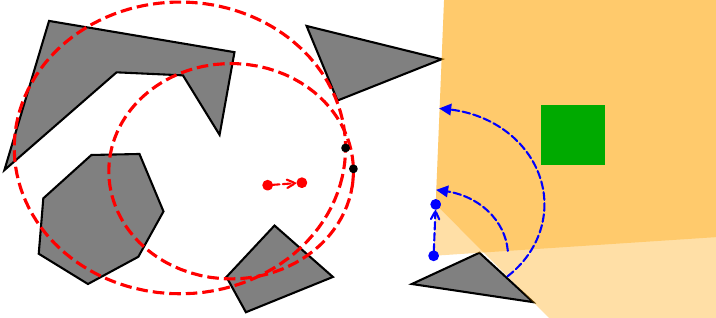}
    \put(-235,75){\scriptsize$(b)$}
    \put(-162,81){\scriptsize$\mathbb{E}$}
    \put(-171,55){\scriptsize$\mathbb{E}'$}
    \put(-149,42){\scriptsize$\bm{x}_{E_j}$}
    \put(-127,42){\scriptsize$\bm{x}_{E_j}'$}
    \put(-109,50){\scriptsize$\bm{x}_I$}
    \put(-106,37){\scriptsize$\bm{x}_I'$}
    \put(-99,15){\scriptsize$\bm{x}_{P_i}$}
    \put(-99,28){\scriptsize$\bm{x}_{P_i}'$}
    \put(-66,30){\scriptsize$\phi$}
    \put(-50,35){\scriptsize$\phi'$}
    \put(-20,65){\scriptsize$\mathcal{R}_i$}
    \put(-34,9){\scriptsize$\mathcal{R}_i'$}
    \caption{Goal-visible pursuit winning. If $P_i$ is goal-visible and the safe distance (i.e., the distance between the region $\mathbb{E}$ (red boundary) and $\goal$) is non-negative, then $P_i$ guarantees to win against $E_j$, which also applies to a pursuit coalition. Let $\mathcal{R}_i$ be a (orange) convex GCP for $\bm{x}_{P_i}$ with the direction angle span $\phi$. The pursuit winning strategy, which ensures the consistent goal-visible property, is as follows. $(a)$ If the closest point $\bm{x}_I$ in $\mathbb{E}$ to $\goal$ is in the angle span $\phi$, then $P_i$ moves towards $\bm{x}_I$. $(b)$ If $\bm{x}_I$ is not in $\phi$, then $P_i$ moves along one edge of $\mathcal{R}_i$. After a small time step, $P_i$ reaches $\bm{x}_{P_i}'$ which is goal-visible and has a convex GCP $\mathcal{R}_i'$ with the angle span $\phi'$. The evader and the closest point are updated to $\bm{x}_{E_j}'$ and $\bm{x}_I'$, respectively.} 
    \label{fig:goal-visible-pursuit-winning}
\end{figure}
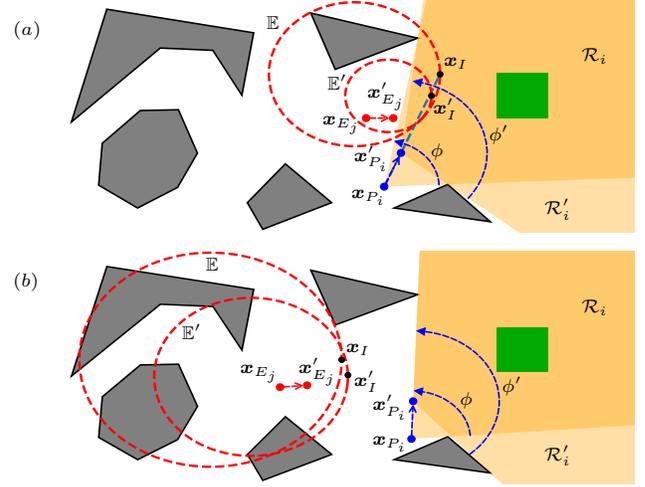

\begin{proof}
    The conclusion follows if we can prove that under the pursuit strategy \eqref{eq:goal-visible-pursuit-strategy}, $\bm{x}_{P_i}(\tau)$ is goal-visible for all $i \in c$ and  the safe distance is non-negative, i.e., $\varrho (X_{cj}(\tau)) \ge 0$ (note that $\bm{x}_{E_j}(\tau) \in \mathbb{E}(c,j)$ before $E_j$ is captured), for all $\tau \ge t$. This holds for $\tau = t$ due to the conditions \ref{itm:goal-visible-pwin-1}) and \ref{itm:goal-visible-pwin-2}).
    
    We first prove the consistent goal-visible property under the strategy \eqref{eq:goal-visible-pursuit-strategy}. Note that $\bm{x}_{P_i}(t)$ is goal-visible and \eqref{eq:goal-visible-pursuit-strategy} indicates that $\theta_i(\tau) \in D_{\mathcal{R}_i(\tau)}(\bm{x}_{P_i}(\tau))$ for all $\tau \ge t$. Then, by \lemaref{lema:consistent-goal-visible-p-strategy}  we obtain that $\bm{x}_{P_i}(\tau)$ is goal-visible for all $\tau \ge t$ and $i \in c$, where a convex GCP $\mathcal{R}_i(\tau)$ exists by the proof of \lemaref{lema:consistent-goal-visible-p-strategy}.

    We show the consistent non-negative safe distance under the strategy \eqref{eq:goal-visible-pursuit-strategy}, i.e., $\varrho (X_{cj}(\tau)) \ge 0$ for all $\tau \ge t$. At time $\tau$, let $(\bm{x}_I(\tau), \bm{x}_G(\tau))$ be the optimal solution to $\mathcal{P}(X_{cj}(\tau))$ in \eqref{eq:convex-pbm-IG} and $\theta_i^{\star}(\tau) = \sigma(\bm{x}_{P_i}(\tau), \bm{x}_I(\tau))$. There are two possible cases.

    Case 1: $\theta_i^{\star}(\tau) \in D_{\mathcal{R}_i(\tau)}(\bm{x}_{P_i}(\tau))$ for all $i \in c$, where one pursuer case is shown in Fig.~\ref{fig:goal-visible-pursuit-winning}(a) and $\phi$ is the angle span of $D_{\mathcal{R}_i(\tau)}(\bm{x}_{P_i}(\tau))$. Then, \eqref{eq:goal-visible-strategy-theta} implies that $\theta_i(\tau) = \theta^{\star}_i(\tau)$ and thus the strategy \eqref{eq:goal-visible-pursuit-strategy} becomes
    \begin{equation}\label{eq:interception-point-strategy}
        \bm{u}_{P_i}(\tau) =  [\cos \theta_i^{\star}(\tau), \sin \theta_i^{\star}(\tau)]^\top = \frac{\bm{x}_I(\tau) - \bm{x}_{P_i}(\tau)}{ \| \bm{x}_I(\tau) - \bm{x}_{P_i}(\tau) \|_2 } \,.
    \end{equation}
    If $\varrho(X_{cj}(\tau)) \ge 0$, then $(\bm{x}_I(\tau), \bm{x}_G(\tau))$ is the unique optimal solution to $\mathcal{P}(X_{cj}(\tau))$ in \eqref{eq:convex-pbm-IG}, as $\mathbb{E}(c,j)$ is strictly convex and $\goal$ is convex. For notation simplicity, we let $d(\bm{x}, \bm{y}) = \| \bm{x} - \bm{y} \|_2$ for $\bm{x}, \bm{y} \in \mathbb{R}^2$, and omit the time argument $\tau$ temporarily here. According to the \emph{Karush-Kuhn-Tucker (KKT) conditions}, the solution $(\bm{x}_I, \bm{x}_G)$ satisfies
    \begin{equation} \label{eq:KKT-conditions}
    \begin{aligned}
        &\bm{0}= \nabla_{\bm{x}} d( \bm{x}_I, \bm{x}_G) + \sum\nolimits_{i \in c} \lambda_i \nabla_{\bm{x}} f_{ij} (\bm{x}_I, X_{ij}) \\
    	&\bm{0}= \nabla_{\bm{y}} d (\bm{x}_I, \bm{x}_G) + \sum\nolimits_{m \in \goalindex} \lambda_m A_m \\
    	& f_{ij} (\bm{x}_I, X_{ij}) \ge 0, \ \lambda_i \leq 0, \ \lambda_i f_{ij}(\bm{x}_I, X_{ij}) = 0,  \ i \in c\\
    	& A_m \bm{x}_G + \bm{b}_m \ge 0,\lambda_m \leq 0, \lambda_m (A_m \bm{x}_G + \bm{b}_m) = 0, m \in \goalindex
    \end{aligned}
    \end{equation}
    where $\lambda_i, \lambda_m \in \mathbb{R}$ are the Lagrange multipliers, and $\nabla_{\bm{x}}$ and $\nabla_{\bm{y}}$ represent the gradient operators with respect to $\bm{x}$ and $\bm{y}$, respectively. The slackness condition on $\bm{x}_G$ in \eqref{eq:KKT-conditions} implies that the index set $\goalindex$ can be classified into two disjoint index sets $\goalindex^{=0}$ and $\goalindex^{>0}$ where
    \begin{equation}\label{eq:slack-condition-G}
    \begin{cases}
    A_m \bm{x}_G + \bm{b}_m = 0, \ \lambda_m \leq 0, \quad \textup{if } m \in \goalindex^{=0} \\
    A_m \bm{x}_G + \bm{b}_m > 0, \ \lambda_m = 0, \quad \textup{if } m \in \goalindex^{>0}.
    \end{cases}
    \end{equation} 
    Then, the speed of  the closure $\mathbb{E}(c,j)$  of the evasion region moving away from $\goal$, i.e., $\dot{\varrho}(X_{cj})$, can be computed as  
    \begin{equation*}
    \begin{aligned}
    	\dot{\varrho}& (X_{cj}) = \dod{}{t}d ( \bm{x}_I, \bm{x}_G ) \\
         &=  \nabla_{\bm{x}}^{\top} d( \bm{x}_I, \bm{x}_G ) \dot{\bm{x}}_I + \nabla_{\bm{y}}^{\top} d ( \bm{x}_I, \bm{x}_G ) \dot{\bm{x}}_G \\
    	& = - \sum\nolimits_{i \in c} \!\! \lambda_i \nabla_{\bm{x}}^{\top} f_{ij} (\bm{x}_I, X_{ij} ) \dot{\bm{x}}_I - \sum\nolimits_{m \in \goalindex} \! \lambda_m A_m^{\top} \dot{\bm{x}}_G \\
        & = - \sum\nolimits_{i \in c} \!\! \lambda_i \nabla_{\bm{x}}^{\top} f_{ij} (\bm{x}_I, X_{ij} ) \dot{\bm{x}}_I - \sum\nolimits_{m \in \goalindex^{=0}} \! \lambda_m A_m^{\top} \dot{\bm{x}}_G \\
    	& = - \sum\nolimits_{i \in c} \!\! \lambda_i \nabla_{\bm{x}}^{\top} f_{ij} (\bm{x}_I, X_{ij} ) \dot{\bm{x}}_I
    \end{aligned}
    \end{equation*}
    where the third and forth equalities follow from \eqref{eq:KKT-conditions} and \eqref{eq:slack-condition-G}, respectively. The last equality follows due to the fact that for $m \in \goalindex^{=0}$, $\bm{x}_G$ is always at the boundary $A_m \bm{x}_G + \bm{b}_m = 0$, and thus $A_m^{\top} \dot{\bm{x}}_G = 0$. Then, by following the similar argument in the proof of \cite[Theorem 3.1]{RY-XD-ZS-YZ-FB:22}, we can obtain that if every pursuer in $P_c$ adopts the strategy \eqref{eq:interception-point-strategy}, then $\dot{\varrho}(X_{cj}) \ge 0$. This guarantees that the safe distance will not strictly decrease and therefore the non-negative property holds for all $\tau \ge t$, as $\varrho(X_{cj}(\tau)) \ge 0$ initially, i.e., when $\tau = t$.

    Case 2: there exist a subset $\bar{c}$ of pursuers in $c$ such that $\theta^{\star}_i(\tau) \notin D_{\mathcal{R}_i(\tau)}(\bm{x}_{P_i}(\tau))$ for all $i \in \bar{c}$, where one pursuer case is in Fig.~\ref{fig:goal-visible-pursuit-winning}(b). The strategy \eqref{eq:goal-visible-pursuit-strategy} indicates that each pursuer in $\bar{c}$ moves along the direction $\theta_i(\tau)$ which is different from $\theta_i^{\star}(\tau)$. Thus, the safe distance may decrease, i.e., $\dot{\varrho}(X_{cj}(\tau)) < 0$, as $\mathbb{E}(c,j)$ may approach $\goal$. However, since $\bm{x}_{P_i}(\tau)$ is goal-visible for all $\tau \ge t$ and $i \in c$ under \eqref{eq:goal-visible-pursuit-strategy}, the direction from $\bm{x}_{P_i}(\tau)$ to any point in $\goal$ is in $D_{\mathcal{R}_i(\tau)}(\bm{x}_{P_i}(\tau))$. Thus, there must exist $\tau' \ge \tau$ such that $\theta^{\star}_i(\tau') \in D_{\mathcal{R}_i(\tau')}(\bm{x}_{P_i}(\tau'))$ for all $i \in c$ before $\mathbb{E}(c,j)$ intersects with the interior of $\goal$, i.e., a negative safe distance. This implies that we go back to Case 1 for which the non-negative safe distance is guaranteed under the strategy \eqref{eq:interception-point-strategy}. Thus, we complete the proof.
\end{proof}

\begin{rek}\label{rek:coalition-reduction}
    Since the safe distance in \defiref{defi:obstacle-free-safe-distance} is defined by ignoring obstacles, the proof of Lemma 3.3 in \cite{RY-XD-ZS-YZ-FB:22} where at most three pursuers are needed to ensure the winning in the three-dimensional space, can be easily adapted to prove that if a pursuit coalition is able to defend the goal region against an evader via the goal-visible pursuit winning, then at most two pursuers in the coalition are necessarily needed.
\end{rek}

Note that the pursuit winning strategy in \thomref{thom:goal-visible-winning} requires a convex GCP $\mathcal{R}_i(\tau)$ for the goal-visible point $\bm{x}_{P_i}(\tau)$ at time $\tau$ for all $i \in c$, which exists by \lemaref{lema:existence-consistency}. If $(\bm{x}_1, \bm{x}_2)$ is a pair of minimum-covering points for $\bm{x}_{P_i}(\tau)$, then as Fig.~\ref{fig:goal-visible-point}(a) indicates, there exists a convex GCP $\mathcal{R}_i(\tau)$ such that
\begin{equation} \label{eq:minimum-covering-direction-range}
   D_{\mathcal{R}_i(\tau)}(\bm{x}_{P_i}(\tau)) = D [\sigma(\bm{x}_{P_i}(\tau), \bm{x}_1), \sigma(\bm{x}_{P_i}(\tau), \bm{x}_2)] .
\end{equation}
However, if $\goal$ is a small region, then the direction range \eqref{eq:minimum-covering-direction-range} has a small angle span and thus under the strategy \eqref{eq:goal-visible-pursuit-strategy}, the pursuer will finally move closely around or in $\goal$. In order to achieve a larger range of movement for the pursuers, we next present a method to construct another class of convex GCPs. Let $\obstaclevertices$ be the set of the vertices of obstacles in $\obstacleset$. We first introduce the following specific vertices. Recall that for two angles $\theta_1 $ and $ \theta_2$, $D[\theta_1, \theta_2]$ is the set of angles from $\theta_1$ to $\theta_2$ in a counterclockwise direction, and $| D[\theta_1, \theta_2]|$ is the angle span of $D[\theta_1, \theta_2]$.

\begin{defi}[First-visible obstacle vertex] \label{defi:first-visible-obstacle-vertex}
    For a goal-visible point $\bm{x} \in \play$, 
    a pair of obstacle vertices $(\bm{y}^L, \bm{y}^U) \in \obstaclevertices^2$ is \emph{first-visible} for $\bm{x}$ if
    \begin{equation}\label{eq:first-visible-obstacle-vertex-L}
    \begin{aligned}
        \bm{y}^L = &  \textup{ argmin}_{\bm{y} \in \obstaclevertices}& & |D[\sigma(\bm{x}, \bm{y}), \sigma(\bm{x}, \bm{x}^L)]| \\
        & \textup{ subject to }& & |D[\sigma(\bm{x}, \bm{y}), \sigma(\bm{x}, \bm{x}^L)]| \leq \pi \\
        & & & |D[\sigma(\bm{x}, \bm{y}), \sigma(\bm{x}, \bm{x}^U)]| \leq \pi \\
        & & & \overline{\bm{x} \bm{y}} \textup{ is obstacle-free}
    \end{aligned}
    \end{equation}
    \begin{equation}\label{eq:first-visible-obstacle-vertex-U}
    \begin{aligned}
        \bm{y}^U = &  \textup{ argmin}_{\bm{y} \in \obstaclevertices}& & |D[\sigma(\bm{x}, \bm{x}^U), \sigma(\bm{x}, \bm{y})]| \\
        & \textup{ subject to }& & |D[\sigma(\bm{x}, \bm{x}^U), \sigma(\bm{x}, \bm{y})]| \leq \pi \\
        & & & |D[\sigma(\bm{x}, \bm{x}^L), \sigma(\bm{x}, \bm{y})]| \leq \pi \\
        & & & \overline{\bm{x} \bm{y}} \textup{ is obstacle-free}
    \end{aligned}
    \end{equation}
    where $(\bm{x}^L, \bm{x}^U)$ is a pair of minimum-covering points for $\bm{x}$. If the problem \eqref{eq:first-visible-obstacle-vertex-L} is infeasible, we define $\bm{y}^L = 2 \bm{x} - \bm{x}^U$, and if the problem \eqref{eq:first-visible-obstacle-vertex-U} is infeasible, we define $\bm{y}^U = 2 \bm{x} - \bm{x}^L$.
\end{defi}

\begin{rek}
    The first-visible obstacle vertex in \defiref{defi:first-visible-obstacle-vertex} is described geometrically in Fig.~\ref{fig:first-observable-obstacle-visible}. 
    Recall that $(\bm{x}^L, \bm{x}^U)$ is a pair of minimum-covering points for $\bm{x}$. Then, $\bm{y}^L$ is the first visible obstacle vertex
    when rotating $\sigma(\bm{x}, \bm{x}^L)$ centered at $\bm{x}$ in a clockwise direction (i.e., span $\phi_1$) before hitting $\sigma(\bm{x}^U, \bm{x})$ (i.e., until span $\phi_2$). If there is no such vertex, $\bm{y}^L = 2 \bm{x} - \bm{x}^U$ is the symmetry point to $\bm{x}^U$ with respect to $\bm{x}$. The point $\bm{y}^U$ is defined similarly but is counterclockwise from $\sigma(\bm{x}, \bm{x}^U)$.    
\end{rek}

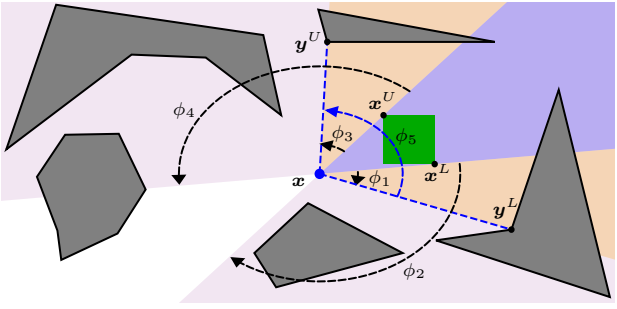
\begin{figure}
    \centering
    \includegraphics[width=0.92\hsize,height=0.45\hsize]{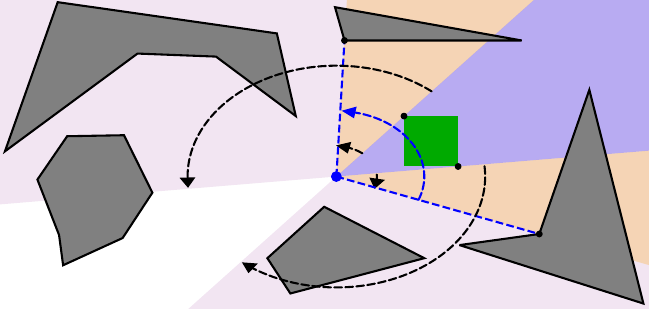}
    \put(-121,96){\scriptsize$\bm{y}^U$}
    \put(-46, 33){\scriptsize$\bm{y}^L$}
    \put(-93,73){\scriptsize$\bm{x}^U$}
    \put(-72,46){\scriptsize$\bm{x}^L$}
    \put(-122,44){\scriptsize$\bm{x}$}
    \put(-93,45){\scriptsize$\phi_1$}
    \put(-80,10){\scriptsize$\phi_2$}
    \put(-107,62){\scriptsize$\phi_3$}
    \put(-167,71){\scriptsize$\phi_4$}
    \put(-83,61){\scriptsize$\phi_5$}
    \caption{First-visible obstacle vertices $\bm{y}^L$ and $\bm{y}^U$. For a goal-visible point $\bm{x}$ with a pair of minimum-covering points $(\bm{x}^L, \bm{x}^U)$, $\bm{y}^L$ is the first visible obstacle vertex
    when rotating $\sigma(\bm{x}, \bm{x}^L)$ in a clockwise direction (i.e., span $\phi_1$) before hitting $\sigma(\bm{x}^U, \bm{x})$ (i.e., until span $\phi_2$). If there is no such vertex, $\bm{y}^L = 2 \bm{x} - \bm{x}^U$ is the symmetry point to $\bm{x}^U$ with respect to $\bm{x}$. The point $\bm{y}^U$ is defined similarly but in a counterclockwise direction from $\sigma(\bm{x}, \bm{x}^U)$.}
    \label{fig:first-observable-obstacle-visible}
\end{figure}

Using the first-visible obstacle vertices, we next construct a class of convex GCPs that achieve a larger range of movement than \eqref{eq:minimum-covering-direction-range} for the pursuers, illustrated in Fig.~\ref{fig:convex_GCP_construction_proof} if $\bm{x}_{P_i} \in \play$.

\begin{thom}[Constructing convex GCPs] \label{thom:construct-convex-GCP}
    For a goal-visible $\bm{x}_{P_i} \in \freespace$, 
    if $\bm{x}_{P_i} \in \goal$, then there exists a convex GCP $\mathcal{R}_i$ such that $D_{\mathcal{R}_i} (\bm{x}_{P_i}) = [0, 2\pi)$; if $\bm{x}_{P_i} \in \play$, then there exists a convex GCP $\mathcal{R}_i$ such that $ D_{\mathcal{R}_i} (\bm{x}_{P_i}) $ is given by
    \begin{equation}
        \begin{cases}
            D[\theta^L, \theta^U], & \textup{if } |D[\theta^L, \theta^U]| \leq \pi \\
            D[\theta^L, \theta^L + \pi] \textup{ or } D[\theta^U - \pi, \theta^U], & \textup{otherwise}
        \end{cases}
    \end{equation}
    where $\theta^L = \sigma( \bm{x}_{P_i}, \bm{y}^L)$, $\theta^U = \sigma( \bm{x}_{P_i}, \bm{y}^U)$ and  $(\bm{y}^L, \bm{y}^U)$ is a pair of first-visible obstacle vertices for  $\bm{x}_{P_i}$.
\end{thom}
\begin{proof}
    Consider a goal-visible point $\bm{x}_{P_i} \in \freespace$. If $\bm{x}_{P_i} \in \goal$, then $\goal$ is a natural convex GCP for $\bm{x}_{P_i}$. Note that $\goal$ does not intersect with the obstacles. We can expand $\goal$ by moving each edge of $\goal$ along its outward normal vector with a small distance and obtain a larger convex GCP $\mathcal{R}_i$ for $\bm{x}_{P_i}$ such that $\bm{x}_{P_i}$ is in the interior of $\mathcal{R}_i$. Therefore, by \rekref{rek:identify-direction-range} we have $D_{\mathcal{R}_i} (\bm{x}_{P_i}) = [0, 2 \pi)$.

    We next construct a convex GCP for $\bm{x}_{P_i} \in \play$, depicted in Fig.~\ref{fig:convex_GCP_construction_proof}. Let $(\bm{x}^L, \bm{x}^U)$ be a pair of minimum-covering points and $(\bm{y}^L, \bm{y}^U)$ a pair of first-visible obstacle vertices for  $\bm{x}_{P_i}$. Also let $\theta^L = \sigma( \bm{x}_{P_i}, \bm{y}^L)$ and $\theta^U = \sigma( \bm{x}_{P_i}, \bm{y}^U)$. 
    There are two cases depending on whether $|D[\theta^L, \theta^U]| \leq \pi$ holds, which will be discussed separately.

    Case 1: $|D[\theta^L, \theta^U]| \leq \pi$, shown in Fig.~\ref{fig:convex_GCP_construction_proof}(a). We first focus on the part on $\theta^L$. We extend the farther boundary of $\goal$ that goes through $\bm{x}^L$, and then obtain an intersection point $\bm{x}''$ with the line segment $\overline{\bm{x}_{P_i} \bm{y}^L}$ (if there is no intersection as in the case for $\theta^U$, we take $\bm{y}^U$). For notation convenience, we define the following polygons via vertices:
    \begin{equation*}
    \begin{aligned}
        & \mathcal{R}_0 : \textup{counterclockwise vertices } \bm{x}_{P_i}, \bm{x}^L, \bm{x}', \bm{x}^U \\
        & \mathcal{R}_1 : \textup{counterclockwise vertices } \bm{x}_{P_i}, \bm{x}'', \bm{x}', \bm{x}^U \\
        & \mathcal{R}_2 : \textup{counterclockwise vertices } \bm{x}_{P_i}, \bm{x}'', \bm{x}^L \,.
    \end{aligned}
    \end{equation*}
    We can use the same argument if there is more than one vertex (i.e., more than $\bm{x}'$) of $\goal$ from $\bm{x}^L$ to $\bm{x}^U$ counterclockwise. 

    \begin{figure}[t]
    \centering
    \includegraphics[width=0.44\hsize]{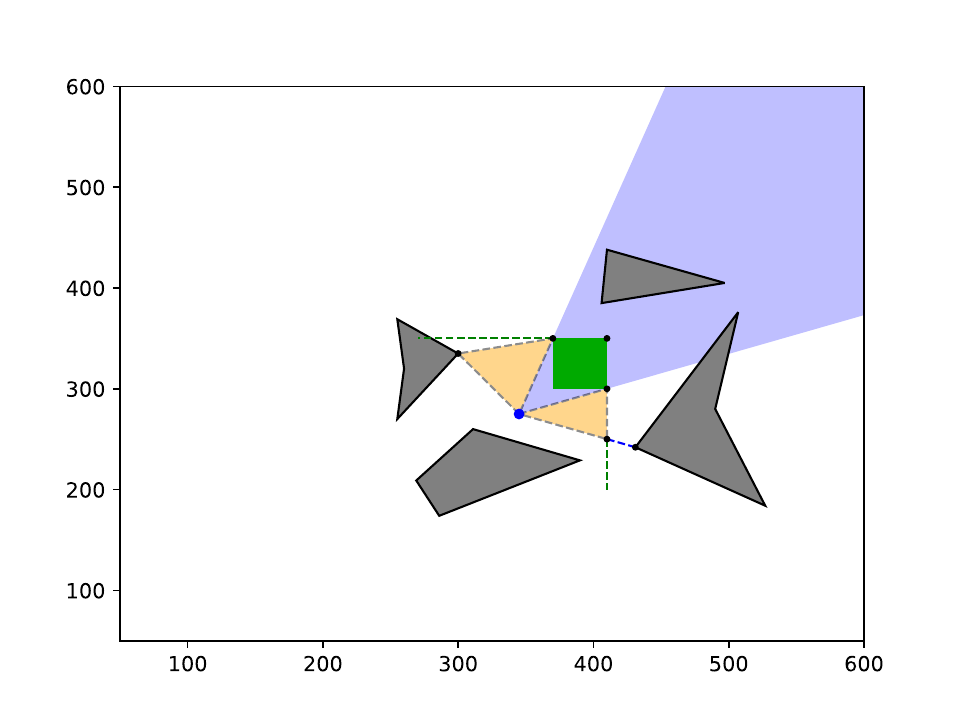} \quad \
    \includegraphics[width=0.44\hsize]{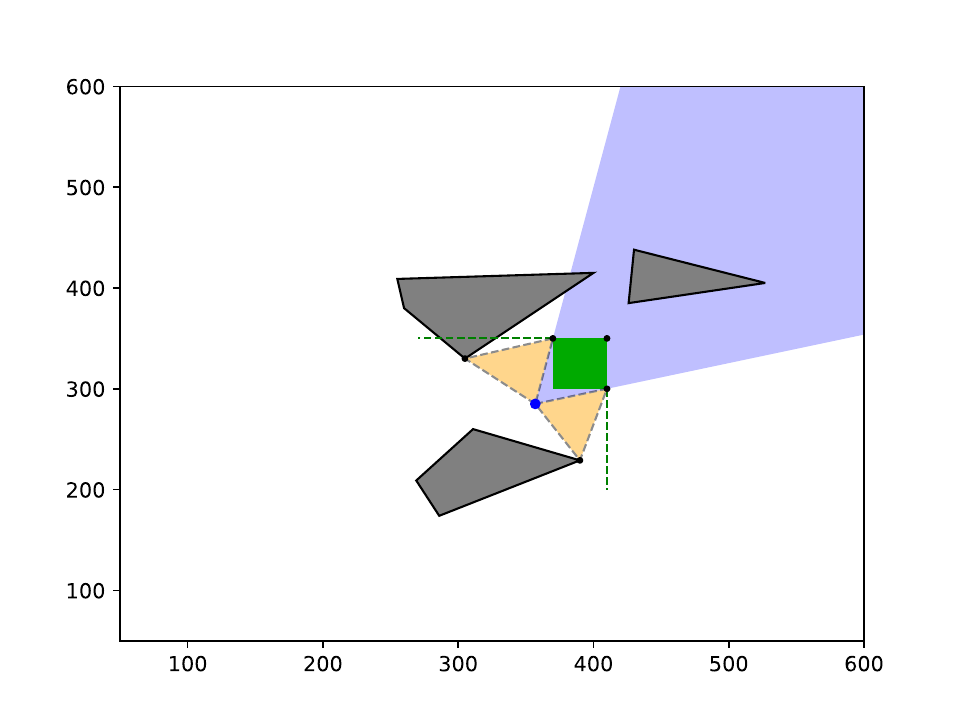}
    \put(-210,26){\scriptsize$\bm{x}_{P_i}$}
    \put(-174,53){\scriptsize$\bm{x}'$}
    \put(-174,34){\scriptsize$\bm{x}^L$}
    \put(-174,25){\scriptsize$\bm{x}''$}
     \put(-199,55){\scriptsize$\bm{x}^U$}
    \put(-172,12){\scriptsize$\bm{y}^L$}
    \put(-225,37){\scriptsize$\bm{y}^U$}
    \put(-78, 28){\scriptsize$\bm{x}_{P_i}$}
    \put(-46,53){\scriptsize$\bm{x}'$}
    \put(-46,33){\scriptsize$\bm{x}^L$}
    \put(-72,55){\scriptsize$\bm{x}^U$}
    \put(-59,12){\scriptsize$\bm{y}^L$}
    \put(-102,42){\scriptsize$\bm{y}^U$}
    \put(-55,0){\scriptsize$(b)$}
    \put(-185,0){\scriptsize$(a)$}
    \caption{Convex GCPs via first-visible obstacle vertices. Let $\theta^L = \sigma( \bm{x}_{P_i}, \bm{y}^L)$ and $\theta^U = \sigma( \bm{x}_{P_i}, \bm{y}^U)$. All orange regions are proved to be obstacle-free. $(a)$ if $|D[\theta^L, \theta^U]| \leq \pi$, then the region (counterclockwise vertices $\bm{x}_{P_i}$, $\bm{x}''$, $\bm{x}'$, $\bm{x}^U$, $\bm{y}^U$) is a convex GCP for $\bm{x}_{P_i}$.
    $(b)$ if $|D[\theta^L, \theta^U]| > \pi$, then two regions (vertices $\bm{x}_{P_i}$, $\bm{y}^L$, $\bm{x}^L$, $\bm{x}'$, $\bm{x}^U$) and (vertices $\bm{x}_{P_i}$, $\bm{x}^L$, $\bm{x}'$, $\bm{x}^U$, $\bm{y}^U$) can be expanded into convex GCPs for $\bm{x}_{P_i}$ such that the direction ranges for $\bm{x}_{P_i}$ are $D[\theta^L, \theta^L + \pi]$ and $D[\theta^U - \pi, \theta^U]$, respectively.}
    \label{fig:convex_GCP_construction_proof}
\end{figure}
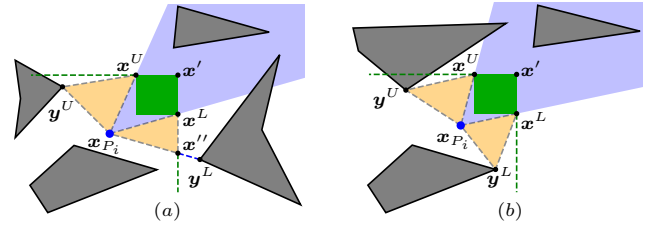

    By \defiref{defi:minimum-covering-points}, $\mathcal{R}_0$ is a natural convex GCP for $\bm{x}_{P_i}$. Then, $\mathcal{R}_1$ is convex as 1) it is constructed by replacing the constraint by the line through $\bm{x}_{P_i}$ and $\bm{x}^L$ with the constraint by the line through $\bm{x}_{P_i}$ and $\bm{x}''$ and 2) the first two constraints in \eqref{eq:first-visible-obstacle-vertex-L} ensure that 
    \[
    \theta^L \in D[\sigma(\bm{x}^U, \bm{x}_{P_i}), \sigma(\bm{x}_{P_i}, \bm{x}^L)] \,.
    \]
    Also note that $\bm{x}_{P_i} \in \mathcal{R}_1$ and $ \goal \subset \mathcal{R}_1$. In order to show that $\mathcal{R}_1$ is a convex GCP for $\bm{x}_{P_i}$, it remains to prove that $\mathcal{R}_1$ is obstacle-free. Since $\mathcal{R}_1 = \mathcal{R}_0 \cup \mathcal{R}_2$ and $\mathcal{R}_0$ is obstacle-free, the problem is reduced to proving that the orange region $\mathcal{R}_2$ is obstacle-free. Suppose that $\mathcal{R}_2$ is not obstacle-free. Then, we have that 1) either $\mathcal{R}_2$ contains at least one obstacle; 2) or there is an obstacle edge penetrating at least one of edges $\overline{\bm{x}'' \bm{x}^L}$ and $\overline{\bm{x}'' \bm{x}_{P_i}}$; 3) or mixture of 1) and 2). For the case 1), there must exist an obstacle vertex in $\mathcal{R}_2$ that has a strictly smaller value for \eqref{eq:first-visible-obstacle-vertex-L}, which contradicts with the fact that $\bm{y}^L$ is the first-visible obstacle vertex. For the case 2), if the obstacle edge only penetrates $\overline{\bm{x}'' \bm{x}^L}$, then it goes back to case 1) as there will be at least one visible obstacle vertex in $\mathcal{R}_2$. If the obstacle edge penetrates $\overline{\bm{x}'' \bm{x}_{P_i}}$, then $\bm{y}^L$ is not visible. For the case 3), it goes back to either 1) or 2). Moreover, 
    $\mathcal{R}_2$ is obstacle-free even if three points $\bm{x}_{P_i}$, $\bm{x}^L$ and $\bm{x}'$ are collinear as we can extend $\overline{\bm{x}^U \bm{x}'}$ via $\bm{x}'$ and prove similarly. 
    Hence, $\mathcal{R}_2$ is obstacle-free and thus $\mathcal{R}_1$ is a convex GCP for $\bm{x}_{P_i}$. 

    Then, for the part on $\theta^U$, we define the following polygons:
    \begin{equation*}
    \begin{aligned}
        & \mathcal{R}_3 : \textup{counterclockwise vertices } \bm{x}_{P_i}, \bm{x}^L, \bm{x}', \bm{x}^U, \bm{y}^U \\
        & \mathcal{R}_4 : \textup{counterclockwise vertices } \bm{x}_{P_i}, \bm{x}'', \bm{x}', \bm{x}^U, \bm{y}^U \,.
    \end{aligned}
    \end{equation*}
    By the similar analysis, we obtain that $\mathcal{R}_3$ is a convex GCP for $\bm{x}_{P_i}$. Since $|D[\theta^L, \theta^U]| \leq \pi$, we have that $\mathcal{R}_4$ is a convex GCP for $\bm{x}_{P_i}$ where $D_{\mathcal{R}_4}(\bm{x}_{P_i}) = D[\theta^L, \theta^U]$.

    Case 2: $|D[\theta^L, \theta^U]| > \pi$, shown in Fig.~\ref{fig:convex_GCP_construction_proof}(b). We define the polygon $\mathcal{R}^5$ whose counterclockwise vertices are $\bm{x}_{P_i}$, $\bm{y}^L$ and $\bm{x}^L$. Following the same argument in the case 1, we obtain that $\mathcal{R}^3$ and $\mathcal{R}^5$ are convex GCPs for $\bm{x}_{P_i}$. Since $\mathcal{R}^5$ is obstacle free and $|D[\theta^L, \theta^U]| > \pi$, then by extending $\overline{\bm{y}^U \bm{x}_{P_i}}$ via $\bm{x}_{P_i}$ we can expand $\mathcal{R}_3$ and construct a convex GCP $\mathcal{R}_6$ for $\bm{x}_{P_i}$ such that $D_{\mathcal{R}_6}(\bm{x}_{P_i}) = D[\theta^U - \pi, \theta^U]$. Similarly, by expanding $\mathcal{R}_5$, we can construct a convex GCP $\mathcal{R}_7$ for $\bm{x}_{P_i}$ such that $D_{\mathcal{R}_7}(\bm{x}_{P_i}) = D[\theta^L, \theta^L + \pi]$, which completes the proof.
\end{proof}

\section{Close-to-goal pursuit winning II: \\ Non-goal-visible case} \label{sec:non-goal-visible}

In this section, we present the second case for the close-to-goal pursuit winning where the pursuers cannot see the whole goal region.  We will construct the pursuit winning region and strategy based on the Euclidean shortest path, reachable region and wavefront in the presence of polygonal obstacles. 


\subsection{Euclidean shortest path, reachable region and wavefront}

We first review three geometric concepts and then present two lemmas on the representations. 

\begin{defi}[Euclidean shortest path, \cite{JH-SS:99}]
    Given two points $\bm{x}_1,\bm{x}_2$ in $\freespace$, a \emph{Euclidean shortest path (ESP)} between them, denoted by $\ESP(\bm{x}_1, \bm{x}_2)$, is an obstacle-free path of minimum total length connecting $\bm{x}_1$ and $\bm{x}_2$.
\end{defi}

\begin{defi}[ESP reachable region]\label{defi:ESP-reachable-region}
     For a point $\bm{x} \in \freespace$, the \emph{ESP reachable region} from $\bm{x}$ in a distance $\ell \ge 0$, denoted by $\ESPregion(\bm{x}, \ell)$, consists of points in $\freespace$ whose ESP distance to $\bm{x}$ is less than or equal to $\ell$.
\end{defi}

\begin{defi}[Wavefront, \cite{JH-SS:99}]\label{defi:wavefront}
     For a point $\bm{x} \in \freespace$, the \emph{wavefront} from $\bm{x}$ in a distance $\ell \ge 0$, denoted by $\wavefront(\bm{x}, \ell)$  (may be empty), consists of points in $\freespace$ whose ESP distance to $\bm{x}$ is $\ell$.
\end{defi}

For visualization, the ESP, ESP reachable region and wavefront are drawn in blue path in Fig.~\ref{fig:ESP-wavefront}(a), grey region and its boundary in Fig.~\ref{fig:ESP-wavefront}(b), respectively. Let $\ESPdist(\bm{x}_1, \bm{x}_2)$ be the length of $\ESP(\bm{x}_1, \bm{x}_2)$. If ESPs are not unique, $\ESP(\bm{x}_1, \bm{x}_2)$ may refer to an arbitrary one. The next two lemmas show that the ESP and wavefront have finite representations.

\begin{figure}[h!]
    \centering
    \includegraphics[width=0.44\hsize]{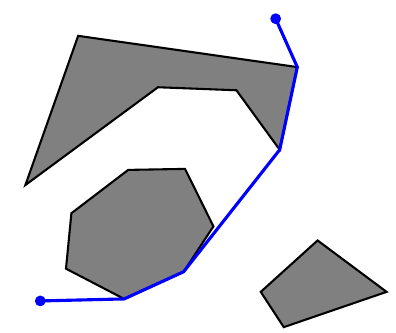} \quad \
    \includegraphics[width=0.44\hsize]{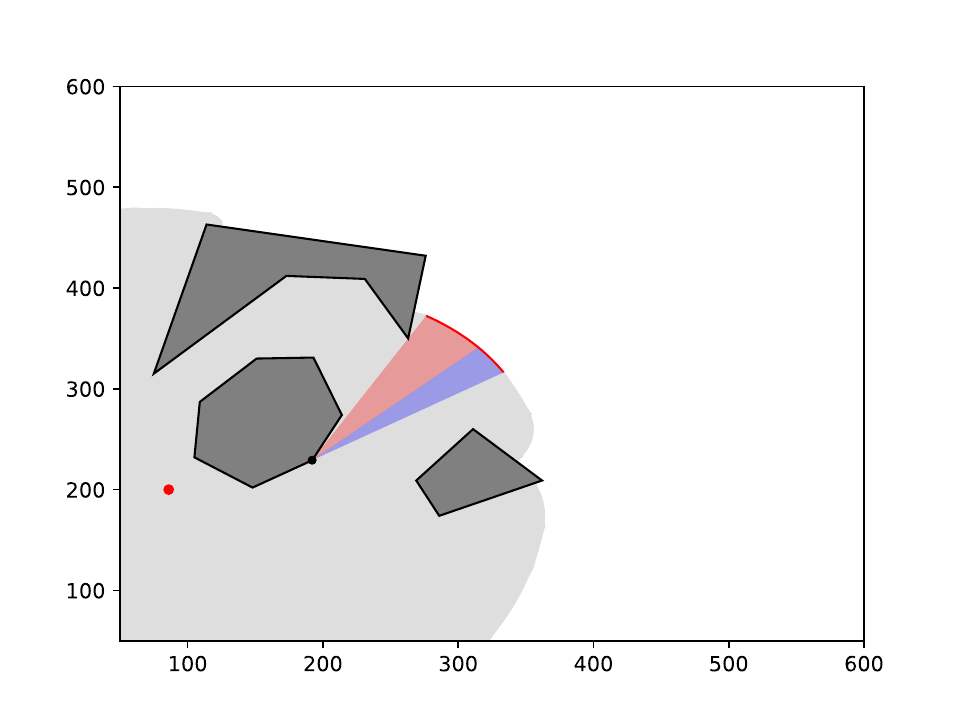}
    \put(-230,2){\scriptsize$\bm{x}_1$}
    \put(-181,46){\scriptsize$\ESP$}
    \put(-103, 2){\scriptsize$\bm{x}$}
    \put(-70,4){\scriptsize$\ESPregion$}
    \put(-158,89){\scriptsize$\bm{x}_2$}
    \put(-10,38){\scriptsize$\wavefront$}
    \put(-57,43){\scriptsize$\sector^k$}
    \put(-61,-6){{\scriptsize$(b)$}}
    \put(-188,-6){\scriptsize$(a)$}
    \caption{Geometric concepts. $(a)$ A (blue) Euclidean shortest path (ESP) $\ESP$ between $\bm{x}_1$ and $\bm{x}_2$. $(b)$ The grey region $\ESPregion$ is the ESP reachable region from $\bm{x}$ in a given distance, and its boundary $\wavefront$ is called wavefront. The wavefront $\wavefront$ is a set of circular wavelets and each (red) wavelet induces a circular sector $\sector^k$ centered on an obstacle vertex or $\bm{x}$. {The red and blue regions are two circular sectors centered on the same vertex induced by two adjacent wavelets, respectively. Notably, the division of $\wavefront$ into wavelets is based on the ESP algorithm in \cite{JH-SS:99}.}}
    \label{fig:ESP-wavefront}
\end{figure}


\begin{lema}[ESP Representation, \cite{JH-SS:99}]\label{lema:ESP-representation}
    An ESP $\ESP(\bm{x}_1, \bm{x}_2)$ can be represented as an ordered set of obstacle vertices plus starting point $\bm{x}_1$ and ending point $\bm{x}_2$.
\end{lema}

\begin{lema}[Wavefront representation, \cite{JH-SS:99}]\label{lema:wavefront-representation}
    A nonempty wavefront $\wavefront(\bm{x}, \ell)$ can be represented as a set of disjoint paths (called wavelets) with an index set $I$, and each wavelet $k \in I$ is a circular arc centered on an obstacle vertex or $\bm{x}$.
\end{lema}

With the ordered set in \lemaref{lema:ESP-representation}, an ESP can be recovered by sequentially connecting the vertices in this set. By \lemaref{lema:wavefront-representation}, each wavelet can be described by a tuple $(\bm{y}, \epsilon, \theta_1, \theta_2)$ which represents a circular arc centered at an obstacle vertex (or the start point) $\bm{y}$ with radius $\epsilon$ starting from angle $\theta_1$ and ending at angle $\theta_2$ counterclockwise. For instance, {two adjacent wavelets (i.e., two red curves, which correspond to circular sectors with different colors) are highlighted in Fig.~\ref{fig:ESP-wavefront}(b)}. A wavefront is then constructed by connecting all adjacent wavelets (please refer to \cite{JH-SS:99} for details). {Notably, the division of $\wavefront$ into wavelets here is based on the ESP algorithm in \cite{JH-SS:99}.}

\begin{rek}\label{rek:wavelet-less-than-pi}
    In this paper, for a wavelet $(\bm{y}, \epsilon, \theta_1, \theta_2)$, if the angle difference from $\theta_1$ to $\theta_2$ in a counterclockwise direction is greater than $\pi$, we split this wavelet evenly into two shorter wavelets whose angle differences will be not more than $\pi$. Such a split can ensure each wavelet induces a convex circular sector (e.g., the red or blue region in Fig.~\ref{fig:ESP-wavefront}(b)), which will be used to construct the pursuit winning region below.
\end{rek}

\subsection{Non-goal-visible pursuit winning}

Combing the ESP, reachable region, wavefront and the goal-visible pursuit winning in Section~\ref{sec:goal-visible}, we next construct the pursuit winning region and strategy for the pursuer who is not at a goal-visible point. The idea is to first steer the pursuer to a goal-visible point and then check if the pursuit winning can be guaranteed from this point via \thomref{thom:goal-visible-winning}. We generalize the safe distance in order to check the condition $(ii)$ in \thomref{thom:goal-visible-winning}. Let $\obstacleverticesobs \subset \obstaclevertices$ be the set of goal-visible obstacle vertices, for instance, green obstacle vertices in Fig.~\ref{fig:non-goal-visible-strategy} are all goal-visible. First, we generalize the safe distance as follows.

\begin{defi}[Anchored ESP-based safe distance]\label{defi:AESP-safe-distance}
    For a state $X_{ij}$ and a goal-visible obstacle vertex $\bm{s} \in \obstacleverticesobs$, the \emph{anchored ESP-based safe distance} $\varrho_{A}(X_{ij}, \bm{s})$ is defined as the optimal value of the problem
    \begin{equation}\label{eq:AESP-safe-distance}
    \begin{aligned}
        & \underset{\bar{\bm{x}}, \bm{x},\bm{y}\in\mathbb{R}^2}{\textup{minimize}}
    	&& \| \bm{x} - \bm{y} \|_2 \\
    	&\textup{subject to}&& \bar{\bm{x}} \in \ESPregion(\bm{x}_{E_j}, \ESPdist(\bm{x}_{P_i}, \bm{s}) / \alpha_{ij} ) \\
        & & & f_{ij}(\bm{x}, \bm{s}, \bar{\bm{x}}) \ge 0,  A_m \bm{y} + \bm{b}_m \ge 0,  \forall m \in \goalindex.
    \end{aligned}
    \end{equation}
\end{defi}

\begin{rek}\label{rek:anchored-ESP-based-safe-distance}
    The anchored ESP-based safe distance $\varrho_{A}(X_{ij}, \bm{s})$ involves two stages as in Fig.~\ref{fig:non-goal-visible-strategy}: in the first stage, $P_i$ moves to a goal-visible obstacle vertex $\bm{s}$ (called anchor point) along an ESP which takes the time $\ESPdist(\bm{x}_{P_i}, \bm{s}) / v_{P_i}$. The grey region is the ESP reachable region $\ESPregion(\bm{x}_{E_j}, \ESPdist(\bm{x}_{P_i}, \bm{s}) / \alpha_{ij} )$ from $\bm{x}_{E_j}$ in this duration. In the second stage, $\varrho_{A}(X_{ij}, \bm{s})$ is defined as the minimal safe distance when $P_i$ is at $\bm{s}$ while $E_j$ can be any point in the grey region. In the figure, $\bar{\bm{x}}^{\star}$ achieves the minimal safe distance.
\end{rek}

Since $\ESPregion(\bm{x}_{E_j}, \ESPdist(\bm{x}_{P_i}, \bm{s}) / \alpha_{ij} )$ is generally non-convex due to the presence of obstacles, computing the anchored ESP-based safe distance involves solving the nonlinear optimization problem \eqref{eq:AESP-safe-distance}. We next present an approximate, efficient solution to the problem \eqref{eq:AESP-safe-distance}, thus providing a sufficient condition for the non-negative anchored ESP-based safe distance. Recall that the angle difference for any wavelet is not more than $\pi$ by \rekref{rek:wavelet-less-than-pi}. The following concept is required.

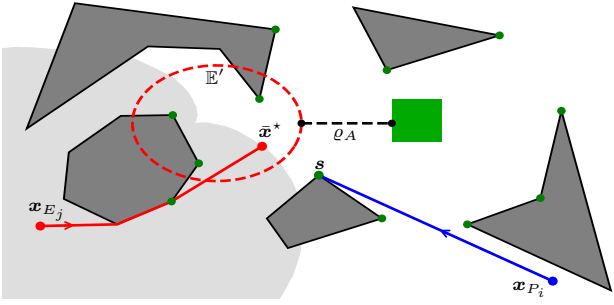
\begin{figure}
    \centering
    \includegraphics[width=0.92\hsize,height=0.45\hsize]{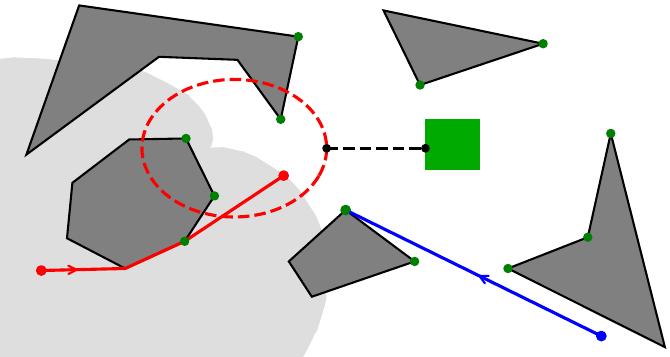}
    \put(-222,33){\scriptsize$\bm{x}_{E_j}$}
    \put(-155, 81){\scriptsize$\mathbb{E}'$}
    \put(-135,61){\scriptsize$\bar{\bm{x}}^{\star}$}
    \put(-39,3){\scriptsize$\bm{x}_{P_i}$}
    \put(-114,49){\scriptsize$\bm{s}$} 
    \put(-107,61){\scriptsize$\varrho_{A}$} 
    \caption{Non-goal-visible pursuit winning. All goal-visible obstacle vertices  $\obstacleverticesobs$ are in green. If $P_i$ is not goal-visible and there exists $\bm{s} \in \obstacleverticesobs$ such that the anchored ESP-based safe distance $\varrho_{A}$ is non-negative, then $P_i$ guarantees to win against $E_j$. The distance $\varrho_{A}$ anchored at $\bm{s}$ is computed as follows. In the first stage, $P_i$ moves to $\bm{s}$ along an (blue) ESP, and the grey region is the ESP reachable region from $\bm{x}_{E_j}$ in this duration. In the second stage, $\varrho_{A}$ is defined as the minimal safe distance when $P_i$ is at $\bm{s}$ and $E_j$ is at any point (say $\bar{\bm{x}}^{\star}$ achieves the minimum) in the grey region. The pursuit winning strategy is to first move along the ESP to $\bm{s}$ and then adopting the goal-visible pursuit winning strategy.}
    \label{fig:non-goal-visible-strategy}
\end{figure}

\begin{defi}[Circular sectors for a wavefront]\label{defi:circular-sector-wavefront}
    For a wavefront $\wavefront(\bm{x}, \ell)$ with an index set $I$, let $\sector^k$ be the circular sector (e.g., red region in Fig.~\ref{fig:ESP-wavefront}(b)) formed by a wavelet $k \in I$, represented by $(\bm{y}, \epsilon, \theta_1, \theta_2)$,  and its center $\bm{y}$, i.e.,
    \begin{equation*}
    \begin{aligned}
        \sector^k = \{ \bm{z} \in \mathbb{R}^2 \mid \; &  \| \bm{z} - \bm{y} \|_2 \leq \epsilon, [-\sin{\theta_1}, \cos{\theta_1}]^{\top} \bm{z} \ge 0 \\
        & [-\sin{\theta_2}, \cos{\theta_2}]^{\top} \bm{z} \leq 0 \} .
    \end{aligned}
    \end{equation*}
    Then, $\sector = \{ \sector^k \mid k \in I \}$ is called the \emph{circular sectors} for  $\wavefront(\bm{x}, \ell)$.
\end{defi}


\begin{lema}[Non-negative anchored ESP-based safe distance]
    For a state $X_{ij}$ and a goal-visible obstacle vertex $\bm{s} \in \obstacleverticesobs$, let $I$ and $\sector$ be the index set and circular sectors for the wavefront $\wavefront(\bm{x}_{E_j}, \ESPdist(\bm{x}_{P_i}, \bm{s}) / \alpha_{ij} )$, respectively. If 
    \begin{enumerate}
        \item $\alpha_{ij} \ESPdist(\bm{x}_{E_j}, \bm{x}) \ge \ESPdist(\bm{x}_{P_i}, \bm{s}) $ for all $\bm{x} \in \goalvertices$; \label{itm:non-negative-anchored-safe-distance-1}
        
        \item $\min_{k \in I} J^{\star}_k \ge 0$; \label{itm:non-negative-anchored-safe-distance-2}
    \end{enumerate}
     where let $d_k = \max_{\bm{z} \in \sector^k} \| \bm{z} - \bm{s} \|_2 $ for a wavelet $k \in I$ and  $J^{\star}_k$ is the optimal value to the convex optimization problem
    \begin{equation}\label{eq:convex-pbm-AESP}
    \begin{aligned}
    & \underset{\bar{\bm{x}}, \bm{x},\bm{y}\in\mathbb{R}^2}{\textup{minimize}}
	&& \| \bm{x} - \bm{y} \|_2 \\
	&\textup{subject to}&& \bar{\bm{x}} \in \sector^k, \; A_m \bm{y} + \bm{b}_m \ge 0,  \forall m \in \goalindex \\
    & & & \| \bm{x} - \bar{\bm{x}} \|_2 \leq (d_k - r_i) / ( \alpha_{ij} - 1) \\
    & & &  \| \bm{x} - \bm{s} \|_2 \leq ( \alpha_{ij} d_k - r_i) / (\alpha_{ij} - 1)
    \end{aligned}
    \end{equation}
    then $\varrho_{A}(X_{ij}, \bm{s}) \ge 0$.
\end{lema}
\begin{proof}
    The condition \ref{itm:non-negative-anchored-safe-distance-1}) indicates that $E_j$ has not yet reached any vertex of $\goal$ when $P_i$ reaches $\bm{s}$. This implies that for each circular sector $\sector^k \in \sector$, the associated center's ESP distance to $\bm{x}_{E_j}$ is not more than $\ESPdist(\bm{x}_{P_i}, \bm{s}) / \alpha_{ij}$.

    Since $\sector^k$ is a convex set and the other constraints in \eqref{eq:convex-pbm-AESP} are convex, then \eqref{eq:convex-pbm-AESP} is a convex optimization problem. 
    
    Using Definitions \ref{defi:ESP-reachable-region}, \ref{defi:wavefront} and \ref{defi:circular-sector-wavefront} and \lemaref{lema:wavefront-representation}, we have $\cup_{k \in I} \sector^k \subset \ESPregion(\bm{x}_{E_j}, \ESPdist(\bm{x}_{P_i}, \bm{s}) / \alpha_{ij} )$ . Therefore, \defiref{defi:AESP-safe-distance} implies that the lemma holds if the last two constraints in \eqref{eq:convex-pbm-AESP} are looser than the constraint $f_{ij}(\bm{x}, \bm{s}, \bar{\bm{x}}) \ge 0$ in \eqref{eq:AESP-safe-distance}.

    We relax $f_{ij}(\bm{x}, \bm{s}, \bar{\bm{x}}) \ge 0$ into two looser constraints:
    \begin{align*}
        & \| \bm{x} - \bar{\bm{x}} \|_2 + \| \bar{\bm{x}} - \bm{s} \|_2 - \alpha_{ij} \| \bm{x} - \bar{\bm{x}} \|_2 - r_i \ge 0 \Rightarrow \\
        & \| \bm{x} - \bar{\bm{x}} \|_2 \leq (\| \bar{\bm{x}} - \bm{s} \|_2 - r_i)/(\alpha_{ij} - 1) {\leq} (d_k - r_i)/(\alpha_{ij} - 1)
    \end{align*}
    and 
    \begin{equation*}
    \begin{aligned}
        & \| \bm{x} - \bm{s} \|_2 - \alpha_{ij} (\| \bm{x} - \bm{s} \|_2 - \| \bm{s} - \bar{\bm{x}} \|_2 ) - r_i \ge 0 \Rightarrow \\
        & \| \bm{x} - \bm{s} \|_2 \leq \frac{\alpha_{ij} \| \bm{s} - \bar{\bm{x}} \|_2 - r_i}{\alpha_{ij} - 1} \leq \frac{\alpha_{ij}d_k - r_i}{\alpha_{ij} - 1}
    \end{aligned}
    \end{equation*}
    which thus completes the proof.
\end{proof}



Utilising all these ESP-based concepts, we next present the pursuit winning region and strategy for the non-goal-visible case of the close-to-goal pursuit winning, illustrated in Fig.~\ref{fig:non-goal-visible-strategy}.


\begin{thom}[Non-goal-visible pursuit winning] \label{thom:non-goal-pursuit-wining}
    At time $t$, if the positions $X_{ij}$ of $P_i$ and $E_j$ are such that 
    \begin{enumerate}
        \item $\bm{x}_{P_i}$ is not goal-visible; \label{itm:non-goal-visible-pwin-1}

        \item there exists at least one goal-visible obstacle vertex $\bm{s} \in \obstacleverticesobs$ such that the anchored ESP-based safe distance is non-negative, i.e., $\varrho_{A}(X_{ij}, \bm{s}) \ge 0$; \label{itm:non-goal-visible-pwin-2}
    \end{enumerate}
    then $P_i$ can guarantee to win against $E_j$, regardless of $E_j$'s strategy, by using the pursuit strategy that computes $\bm{u}_{P_i}(\tau)$ for $\tau \ge t$ as follows:
\begin{equation}\label{eq:non-goal-visible-pursuit-strategy}
    \begin{cases}
    \textup{Along } \ESP(\bm{x}_{P_i}(t), \bm{s}) & \textup{if } t \leq \tau \leq t + \ESPdist(\bm{x}_{P_i}(t), \bm{s}) / v_{P_i} \\
    \eqref{eq:goal-visible-pursuit-strategy} & \textup{if } \tau > t + \ESPdist(\bm{x}_{P_i}(t), \bm{s}) / v_{P_i} \,.
    \end{cases}
\end{equation} 
\end{thom}
\begin{proof}
    A goal-visible obstacle vertex $\bm{s} \in \obstacleverticesobs$ exists such that $\varrho_{A}(X_{ij}(t), \bm{s}) \ge 0$. Then, as \rekref{rek:anchored-ESP-based-safe-distance} states, $P_i$ can reach the vertex $\bm{s}$ along the ESP $\ESP(\bm{x}_{P_i}(t), \bm{s})$ which takes the time $\ESPdist(\bm{x}_{P_i}(t), \bm{s}) / v_{P_i}$. The safe distance is non-negative when $P_i$ reaches $\bm{s}$, regardless of $E_j$' strategy. Thus, since $\bm{s}$ is goal-visible, then by \thomref{thom:goal-visible-winning} the strategy \eqref{eq:goal-visible-pursuit-strategy} can ensure $P_i$'s winning against $E_j$.
\end{proof}

\section{Multiplayer Pursuit Strategy} \label{sec:multiplayer-strategy}

In this section, we propose a multiplayer pursuit strategy by fusing the subgame outcomes in Sections \ref{sec:onsite-pursuit-winning}, \ref{sec:goal-visible} and \ref{sec:non-goal-visible} with hierarchical optimal task allocation to ensure a lower bound on the number of defeated evaders and improve the bound continually. Before that, we first present an evasion winning region and strategy which can conversely help the pursuit team to distribute its team members to the evaders more efficiently.



\subsection{Evasion winning}

Based on the ESP, we construct an evasion winning region and strategy which will be used in the following matchings.

\begin{lema}[ESP-based evasion winning]\label{lema:evasion-winning}
If the positions $X_{ij}$ of $P_i$ and $E_j$ are such that there exists $\bm{x} \in \goal$ satisfying 
\begin{equation} \label{eq:evader-win-condition}
    \ESPdist(\bm{x}_{P_i}, \bm{x}) - r_i > \alpha_{ij} \ESPdist(\bm{x}_{E_j}, \bm{x})
\end{equation}
then $E_j$ can guarantee to reach $\goal$ before being captured by $P_i$ by moving along $\ESP(\bm{x}_{E_j}, \bm{x})$, regardless of $P_i$' strategy.  
\end{lema}
\begin{proof}
Suppose that at time $t$, the positions of $P_i$ and $E_j$ are $X_{ij}$. The conclusion follows if we can prove that the distance between two players is always greater than $r_i$ when $E_j$ moves along $\ESP(\bm{x}_{E_j}(t), \bm{x})$. It implies that $\| \bm{x}_{P_i}(\tau) - \bm{x}_{E_j}(\tau) \|_2 > r_i$ for all $t \leq \tau \leq t+ \ESPdist(\bm{x}_{E_j}(t), \bm{x}) / v_{E_j} $. For simplicity, we let $\ESPdist^P = \ESPdist(\bm{x}_{P_i}(t), \bm{x})$ and $\ESPdist^E = \ESPdist(\bm{x}_{E_j}(t), \bm{x})$.

Suppose that there exists a pursuit strategy $\bm{u}_{P_i}'$ such that $E_j$ is captured by $P_i$, i.e., $\| \bm{x}_{P_i}(\tau') - \bm{x}_{E_j}(\tau') \|_2 \leq r_i$, at some time $t \leq \tau' \leq \ESPdist^E / v_{E_j} + t$. We can construct an obstacle-free path $\mathcal{P}'$ along which $P_i$ can reach $\bm{x}$: first, move to $\bm{x}_{P_i}(\tau')$ under $\bm{u}_{P_i}'$; then, move to $\bm{x}_{E_j}(\tau')$ directly (recall that this path is obstacle-free due to the visual capture condition); finally, move along part of $\ESP(\bm{x}_{E_j}(t), \bm{x})$ to reach $\bm{x}$ as $\bm{x}_{E_j}(\tau')$ is on the path $\ESP(\bm{x}_{E_j}(t), \bm{x})$. Let $t'$ be the time of $P_i$ reaching $\bm{x}$ along the path $\mathcal{P}'$. Then we have $t' \leq \tau' + (r_i + \ESPdist^E - v_{E_j} (\tau' - t) ) / v_{P_i}$. Thus, the length of $\mathcal{P}'$, denoted by $d(\mathcal{P}')$, satisfies
\begin{equation*}
\begin{aligned}
   &  d(\mathcal{P}') -  \ESPdist^P = v_{P_i} (t' - t) - \ESPdist^P\\
   & \leq v_{P_i}(\tau' - t) + r_i + \ESPdist^E - v_{E_j} (\tau' - t) - \ESPdist^P \\
   & < (1 - 1/\alpha_{ij}) v_{P_i}(\tau' - t) + r_i + (\ESPdist^P - r_i) / \alpha_{ij} - \ESPdist^P \\
   & = (1 - 1/\alpha_{ij}) ( v_{P_i}(\tau' - t) + r_i -  \ESPdist^P )  \\
   & \leq  (1 - 1/\alpha_{ij}) ( v_{P_i} \ESPdist^E / v_{E_j} + r_i -  \ESPdist^P ) < 0,
\end{aligned}
\end{equation*}
where the second and fourth inequalities are due to \eqref{eq:evader-win-condition}, which is a contradiction as  $\ESPdist^P$ is the ESP from $\bm{x}_{P_i}$ and $\bm{x}$.
\end{proof}

As a sufficient and efficient method, we practically check the evasion winning condition \eqref{eq:evader-win-condition} over finitely many critical points in $\goal$, for instance, the vertices of $\goal$.

\subsection{Multiplayer onsite and close-to-goal pursuit strategy}
We next present the \emph{multiplayer onsite and close-to-goal (MOCG) pursuit strategy} in \algoref{alg:multiplayer-strategy}. In each iteration, the strategy first generates a hierarchical task allocation between pursuit coalitions and evaders to maximize the number of defeated evaders (a lower bound). Then, the strategy computes the pursuit strategies that can ensure the lower bound.  The task allocation  is dynamic and will change over time (therefore, the pursuit strategies may also change), if a new task allocation  guaranteeing to defeat more evaders is generated as the game evolves. The hierarchical task allocation involves four matchings which are generated sequentially at each iteration. 

\begin{algorithm}
\caption{MOCG pursuit strategy}
\textbf{Initialize:} $\{ \bm{x}_{P_i}^0\}_{P_i \in \pteam}$, $\{ \bm{x}_{E_j}^0 \}_{E_j \in \eteam}$

\begin{algorithmic}[1]
\State $\bm{x}_{P_i} \leftarrow \bm{x}_{P_i}^0$, $\bm{x}_{E_j} \leftarrow \bm{x}_{E_j}^0$ for all $P_i \in \pteam$ and $E_j \in \eteam$
\State $\mathcal{V}_P \leftarrow [\pteam]^2$, $\mathcal{V}_E \leftarrow \eteam$
\State $\eteam_{\textup{capture}} \leftarrow \emptyset$, $M_{\textup{defeat}} \leftarrow \emptyset$
\Repeat
\For{$P_i \in \pteam$, $E_j \in \mathcal{V}_E$}
\State $T_{ij} \leftarrow \text{Check\_pursuit\_winning}(X_{ij})$
\State Add $e_{ij}$ to $\mathcal{E}$ if $T_{ij} \ge 1$
\EndFor
\For{$P_c \in \mathcal{V}_P$, $E_j \in \mathcal{V}_E$}
\If{$|c| = 2$ and $e_{ij} \notin \mathcal{E}$ for all $i \in c$}
\State $T_{cj} \leftarrow \text{Check\_pursuit\_winning}(X_{cj})$
\State Add $e_{cj}$ to $\mathcal{E}$ if $T_{cj} = 2$
\EndIf
\EndFor
\State $M^{\star}_1 \leftarrow $ solve the BIP \eqref{eq:BIP-captured-evaders} for graph $(\mathcal{V}_P \cup \mathcal{V}_E, \mathcal{E})$
\State $M_{\textup{defeat}} \leftarrow M^{\star}_1$ if $|M^{\star}_1| > M_{\textup{defeat}}$
\For{$e_{cj} \in M_{\textup{defeat}}$}
\State $\bm{u}_{P_i} \leftarrow$ \lemaref{lema:onsite-pursuit-strategy-1v1} agst. $E_j$ if $T_{ij} = 1$ ($c=\{i\}$)
\State $\bm{u}_{P_i} \leftarrow$ \eqref{eq:goal-visible-pursuit-strategy} agst. $E_j$ if $T_{cj} = 2$ for all $i \in c$
\State $\bm{u}_{P_i} \leftarrow$ \eqref{eq:non-goal-visible-pursuit-strategy} agst. $E_j$ if $T_{ij} = 3$ ($c=\{i\}$)
\EndFor
\For{unassigned $P_i \in \pteam$}
\State $\bm{u}_{P_i} \leftarrow$ \lemaref{lema:onsite-pursuit-strategy-1v1} agst. an $E_j \in \mathcal{V}_E$ if $T_{ij} = 1$ 
\EndFor
\State $\bar{\mathcal{V}}_P, \bar{\mathcal{V}}_E \leftarrow $ all unmatched $P_i \in \pteam$ and $E_j \in \mathcal{V}_E$
\For{$P_i \in \bar{\mathcal{V}}_P $, $E_j \in \bar{\mathcal{V}}_E$}
\State Add $e_{ij}$ to $\bar{\mathcal{E}}$ if \eqref{eq:evader-win-condition} does not hold
\EndFor
\State $M^{\star}_2 \leftarrow $ a maximum matching for graph $ (\bar{\mathcal{V}}_P \cup \bar{\mathcal{V}}_E, \bar{\mathcal{E}})$
\State $\bm{u}_{P_i} \leftarrow$ \defiref{defi:ESP-based-pure-pursuit} agst. $E_j$ for all $e_{ij} \in M^{\star}_2$
\For{unassigned $P_i \in \pteam$}
\State $\bm{u}_{P_i} \leftarrow$ \defiref{defi:ESP-based-pure-pursuit} agst. closest $E_j$
\EndFor
\State Adopt any strategy for $E_j \in \mathcal{V}_E$
\State Update $\bm{x}_{P_i}, \bm{x}_{E_j}$ with a time step $\Delta$
\State $\eteam_{\textup{capture}} \leftarrow \eteam_{\textup{capture}} \cup \{ \textup{captured evader(s) in } \mathcal{V}_E \}$
\State Remove captured/arrival evaders from $\mathcal{V}_E$ and $M_{\textup{defeat}}$
\Until{$\mathcal{V}_E = \emptyset $}
\end{algorithmic}
\label{alg:multiplayer-strategy}
\end{algorithm}

\begin{algorithm}
\caption{$\text{Check\_pursuit\_winning}(X_{cj})$ }
\begin{algorithmic}[1]
\State $T_{cj} \leftarrow 0$
\If{$X_{cj}$ satisfies \eqref{eq:onsite-winning-condition-1v1} for all $i \in c$}
\State $T_{cj} \leftarrow 1$ \Comment{Onsite pursuit winning}
\ElsIf{$X_{cj}$ meets  conditions in \thomref{thom:goal-visible-winning}}
\State $T_{cj} \leftarrow 2$ \Comment{Visible-goal pursuit winning}
\ElsIf{$|c|=1$ \& $X_{cj}$ meets conditions in \thomref{thom:non-goal-pursuit-wining}}
\State $T_{cj} \leftarrow 3$ \Comment{Non-visible-goal pursuit winning}
\EndIf
\State \Return$T_{cj}$
\end{algorithmic}
\label{alg:check_pursuit_winning}
\end{algorithm}

\startpara{Capture matching} 
Firstly, we generate a capture matching in which each matched evader can be defeated by the assigned pursuer(s). By Remarks \ref{rek:onsite-winning} and \ref{rek:coalition-reduction} and \thomref{thom:non-goal-pursuit-wining}, the capture matching is simplified as we only need to consider all pursuit coalitions of size no more than two. Let $\mathcal{G} = (\mathcal{V}_P \cup \mathcal{V}_E, \mathcal{E})$ be an undirected bipartite graph with two vertex sets $\mathcal{V}_P$ and $\mathcal{V}_E$, and a set of edges $\mathcal{E}$. In our problem, $\mathcal{V}_P$ is the set of all nonempty pursuit coalitions of size less than or equal to two, and $\mathcal{V}_E$ the set of evaders. The edge connecting vertex $P_c \in \mathcal{V}_P$ and vertex $E_j \in \mathcal{V}_E$ is denoted by $e_{cj}$. For a state $X_{cj}$, we check the pursuit winning via $\text{Check\_pursuit\_winning}(X_{cj})$ (\algoref{alg:check_pursuit_winning}). We add $e_{cj}$ into $\mathcal{E}$ if and only if the returned value satisfies $T_{cj} \ge 1$, that is, $P_c$ guarantees to win against $E_j$ through the onsite ($T_{cj} = 1$), goal-visible ($T_{cj} = 2$) or non-goal-visible ($T_{cj} = 3$) pursuit winnings. Let $\mathcal{C} = ( \mathcal{E}, \bar{\mathcal{E}})$ be a conflict graph of $\mathcal{G}$, where each vertex in $\mathcal{C}$ corresponds uniquely to an edge in $\mathcal{G}$. An edge $\bar{e} \in \bar{\mathcal{E}}$ if and only if two vertexes (two edges in $\mathcal{G}$) connected by $\bar{e}$ involve at least one common pursuer.

The problem of maximizing the number of defeated evaders can be formulated as a binary integer program (BIP):
\begin{equation}\label{eq:BIP-captured-evaders}
\begin{aligned}
    & \underset{a_{cj}, a_{pq} \in \{ 0, 1\} }{\textup{maximize}}
	&& \sum\nolimits_{e_{cj} \in \mathcal{E} } a_{cj} \\
	&\textup{ subject to}&& \sum\nolimits_{P_c \in \mathcal{V}_P } a_{cj} \leq 1, \quad \forall E_j \in \mathcal{V}_E \\
	& & & \sum\nolimits_{E_j \in \mathcal{V}_E} a_{cj} \leq 1, \quad \forall P_c \in \mathcal{V}_P \\
	& & & a_{cj} + a_{pq} \leq 1, \quad \forall (e_{cj}, e_{pq}) \in \bar{\mathcal{E}}
\end{aligned}
\end{equation}
where $a_{cj} = 1$ means the allocation of $P_c$ to capture  $E_j$, and $a_{cj} = 0$ means no assignment. Let $\bm{a}^{\star}$ be an optimal solution to \eqref{eq:BIP-captured-evaders}. Then, we obtain a maximum matching $M^{\star}_1 \subset \mathcal{E}$ of $\mathcal{G}$ subject to the constraint $\mathcal{C}$, where $e_{cj} \in M^{\star}_1$ if and only if $a^{\star}_{cj} = 1$. For each assignment $(P_c, E_j) \in M^{\star}_1$, the pursuer(s) in $P_c$ adopts the strategy in \lemaref{lema:onsite-winning-nv1}, Theorems \ref{thom:goal-visible-winning} or \ref{thom:non-goal-pursuit-wining} depending on the value of $T_{cj}$ to capture $E_j$ (lines 5-17).

\startpara{Enhanced matching} We then generate an enhanced matching to improve the performance of the onsite matchings in $M^{\star}_1$. As \rekref{rek:onsite-winning} indicates that, the more pursuers that can ensure the onsite pursuit winning against an evader are tasked to the evader, the smaller {or equal} region in which the evader will be captured. Therefore, we assign each unassigned pursuer to an evader it can win against via the onsite winning (lines 18-19).

\startpara{Non-dominated matching} We next generate a non-dominated matching in which no pursuer is assigned to an evader that it cannot win against. For the remaining unmatched pursuers $\bar{\mathcal{V}}_P$ and evaders $\bar{\mathcal{V}}_E$, let $\bar{\mathcal{G}} = (\bar{\mathcal{V}}_P \cup \bar{\mathcal{V}}_E, \bar{\mathcal{E}})$ be an undirected bipartite graph, where $e_{ij} \in \bar{\mathcal{E}}$ if and only if $E_j$ cannot ensure the evasion winning in \lemaref{lema:evasion-winning}. A non-dominated matching $M^{\star}_2$ is then a maximum bipartite matching of $\bar{\mathcal{G}}$ which can be computed in polynomial time (e.g., via maximum network flow \cite{LRFJ-DRF:15}). Each pursuer adopts the following ESP-pure strategy to pursue the matched evader (lines 20-24).

\begin{defi}[ESP-pure pursuit strategy]\label{defi:ESP-based-pure-pursuit}
    For the positions $X_{ij}$ of $P_i$ and $E_j$, if $\bm{u}_{P_i}$ is along the ESP $\ESP(\bm{x}_{P_i}, \bm{x}_{E_j})$ at $\bm{x}_{P_i}$, then $\bm{u}_{P_i}$ is called the \emph{ESP-pure pursuit strategy}.
\end{defi}

\startpara{Closest matching} We finally generate a closest matching for the remaining unassigned pursuers. If these pursuers exist, then it means that either 1) all evaders have been matched or 2) the remaining evaders can win against them. Therefore, a closest matching is as follows: each remaining unassigned pursuer pursues the closest unmatched evader if present and the closest evader otherwise, via the ESP-pure pursuit strategy (lines 25-26).


We have the following guarantee on the number of defeated evaders with the MOCG pursuit strategy.

\begin{thom}[Increasing lower bound for guaranteed defeated evaders]
Using the MOCG pursuit strategy, the pursuit team $\pteam$ guarantees to win against at least $| \eteam_{\textup{capture}} | + | M_{\textup{defeat}} |$ evaders simultaneously at each iteration, regardless of the evasion team $\eteam$'s strategy. Moreover, $| \eteam_{\textup{capture}} | + | M_{\textup{defeat}} |$ is increasing with iterations.
\end{thom}
\begin{proof}
    The conclusion directly follows from Theorems \ref{thom:onsite-pursuit-winning}, \ref{thom:goal-visible-winning} and \ref{thom:non-goal-pursuit-wining}, and the fact that the matching will change only if more evaders can be defeated in the new matching.
\end{proof}

\section{Simulations} \label{sec:simulation}
\def\height{0.23}

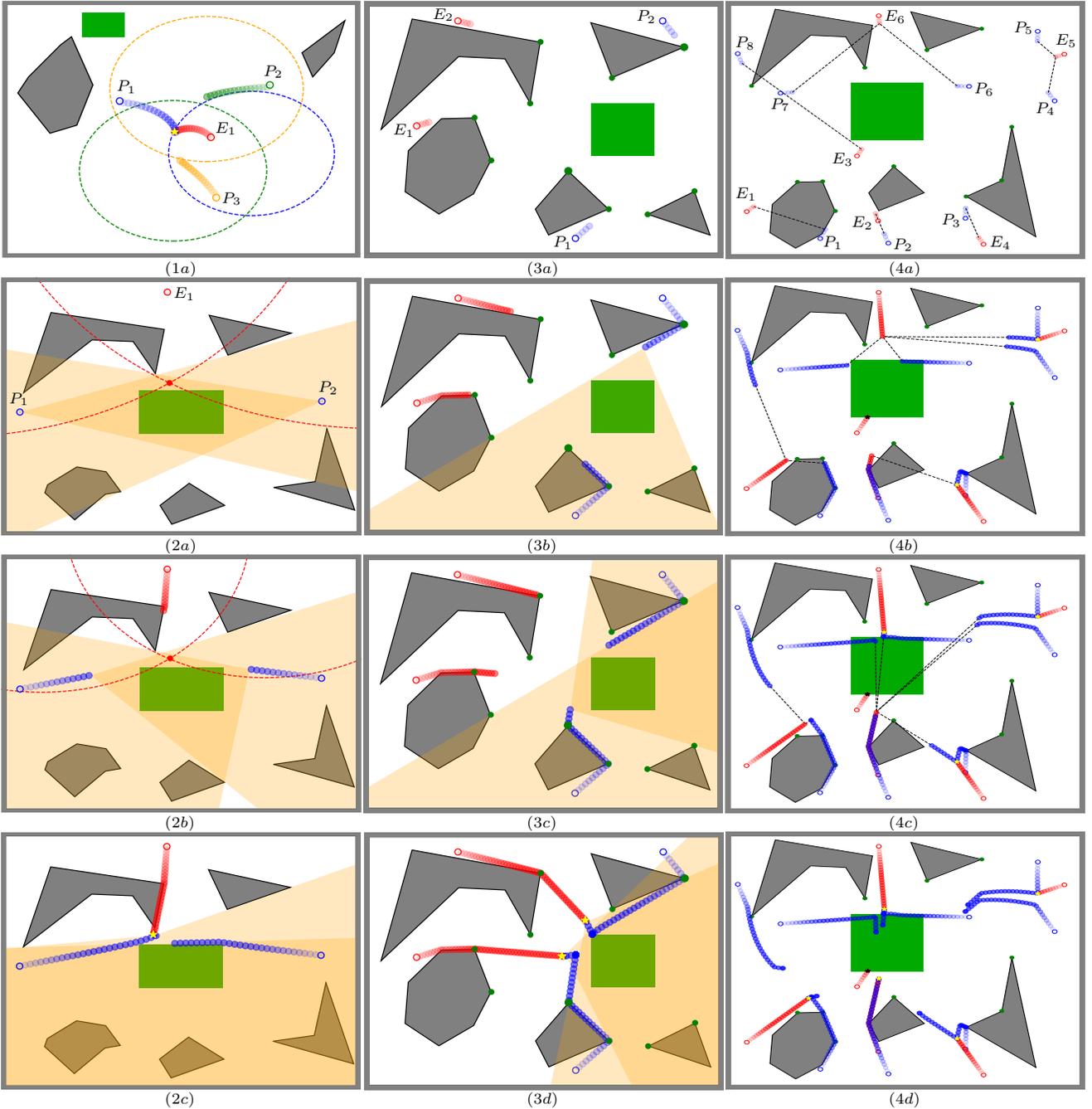
\begin{figure*}
    \centering
    \includegraphics[width=0.32\hsize,height=\height\hsize]{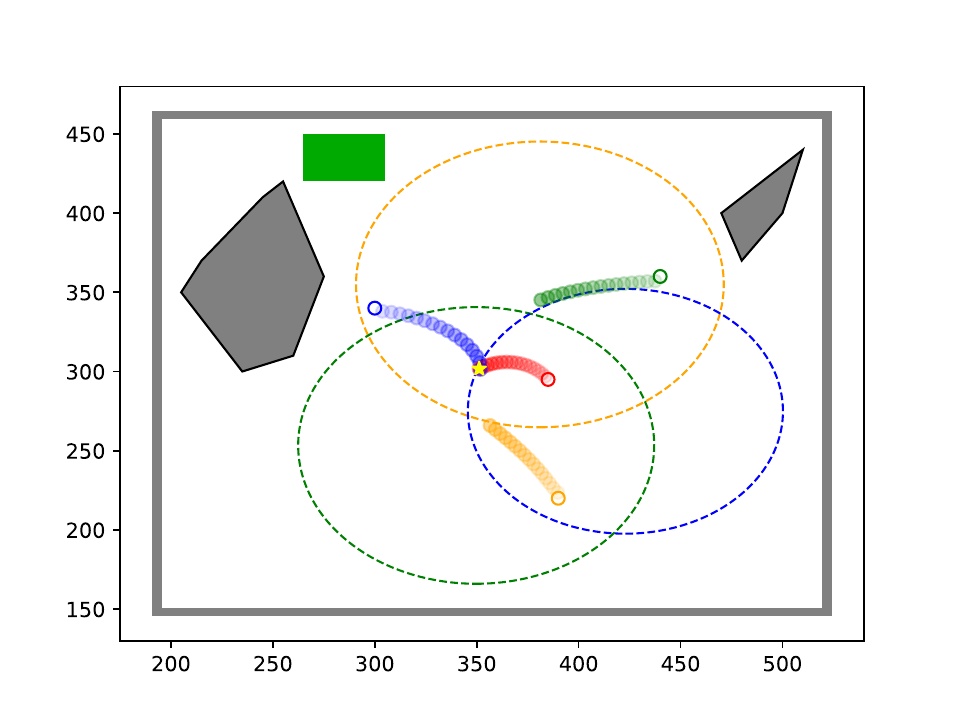}
    \put(-90,-6){\scriptsize$(1a)$}
    \put(-112,77){\scriptsize$P_1$}
    \put(-45,84){\scriptsize$P_2$}
    \put(-64,26){\scriptsize$P_3$}
    \put(-67,58){\scriptsize$E_1$}
    \includegraphics[width=0.32\hsize,height=\height\hsize]{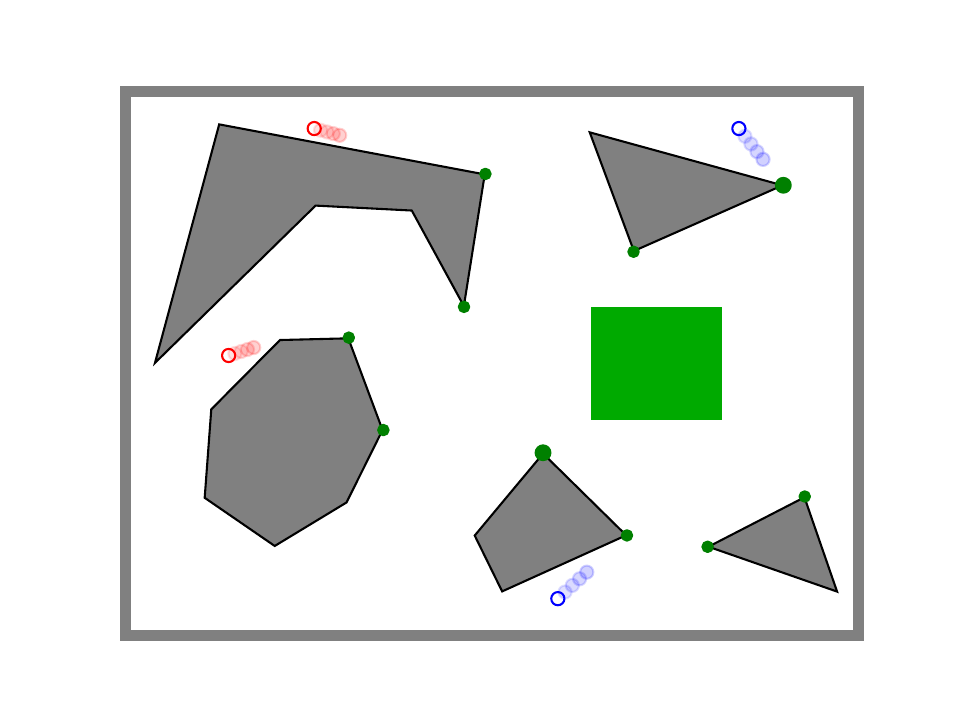}
    \put(-90,-6){\scriptsize$(3a)$} 
    \put(-78,6){\scriptsize$P_1$}
    \put(-40,109){\scriptsize$P_2$}
    \put(-133,110){\scriptsize$E_2$}
    \put(-151,57){\scriptsize$E_1$}
    \includegraphics[width=0.32\hsize,height=\height\hsize]
    {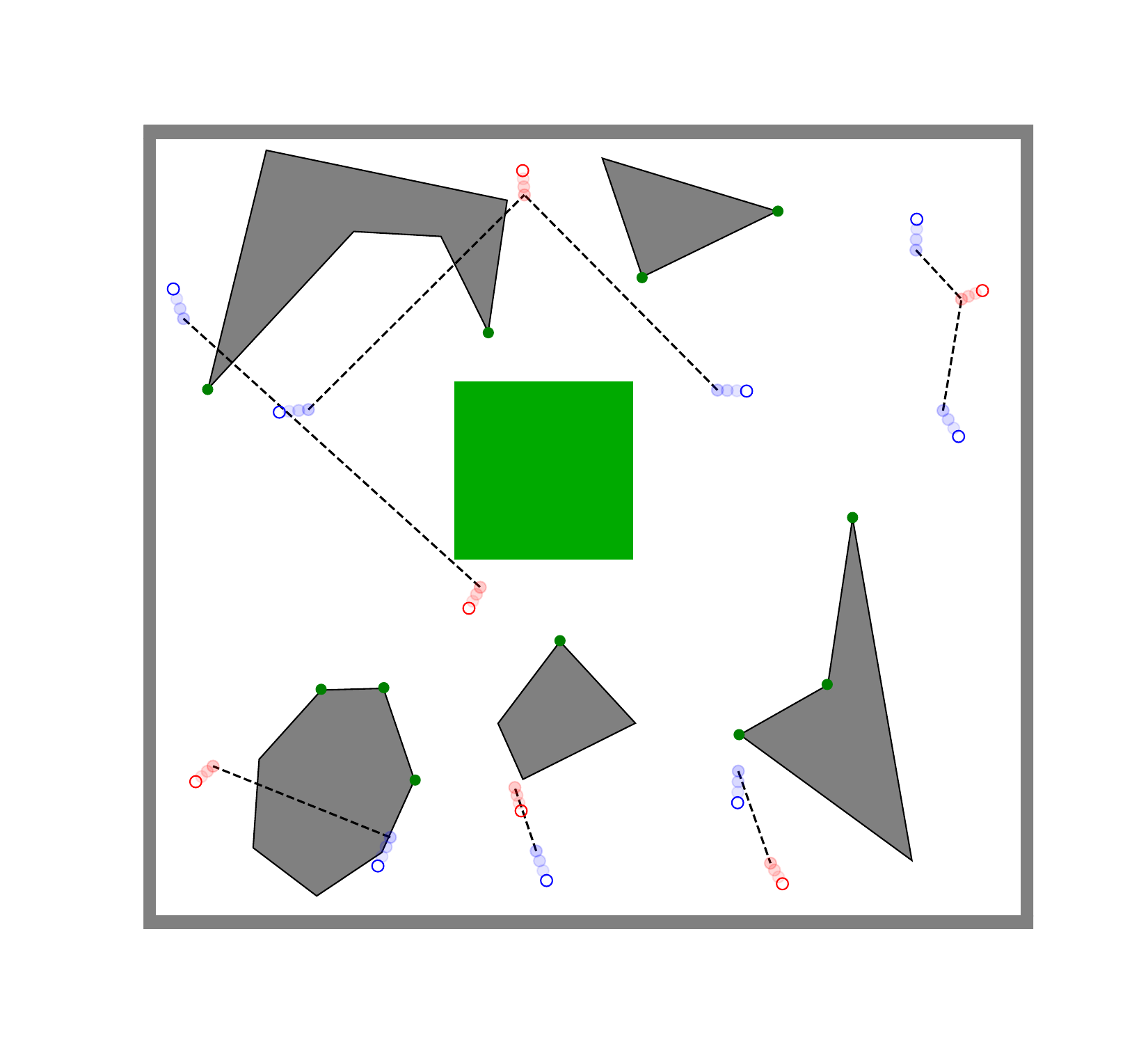}
    \put(-90,-6){\scriptsize$(4a)$} 
    \put(-160,28){\scriptsize$E_1$}
    \put(-107,15){\scriptsize$E_2$}
    \put(-116,44){\scriptsize$E_3$}
    \put(-44,6){\scriptsize$E_4$}
    \put(-14,98){\scriptsize$E_5$}
    \put(-92,109){\scriptsize$E_6$}
    \put(-120,6){\scriptsize$P_1$}
    \put(-88,5){\scriptsize$P_2$}
    \put(-66,16){\scriptsize$P_3$}
    \put(-23,67){\scriptsize$P_4$}
    \put(-33,104){\scriptsize$P_5$}
    \put(-51,76){\scriptsize$P_6$}
    \put(-144,69){\scriptsize$P_7$}
    \put(-160,97){\scriptsize$P_8$}
    
    \includegraphics[width=0.32\hsize,height=\height\hsize]{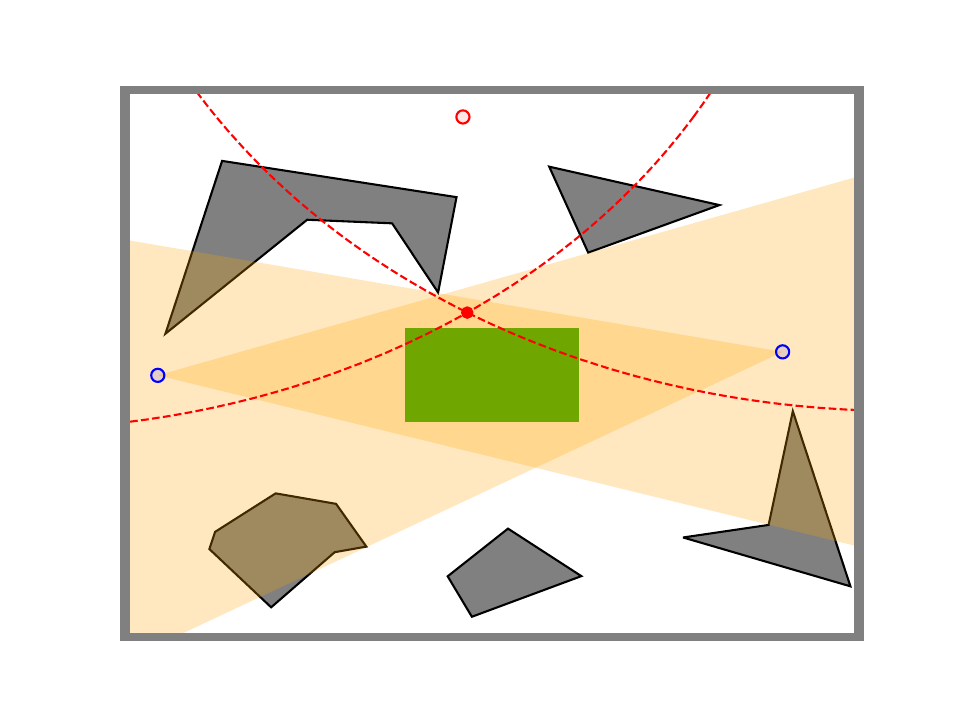}
    \put(-90,-6){\scriptsize$(2a)$} 
    \put(-161,61){\scriptsize$P_1$}
    \put(-20,66){\scriptsize$P_2$}
    \put(-86,109){\scriptsize$E_1$}
    \includegraphics[width=0.32\hsize,height=\height\hsize]{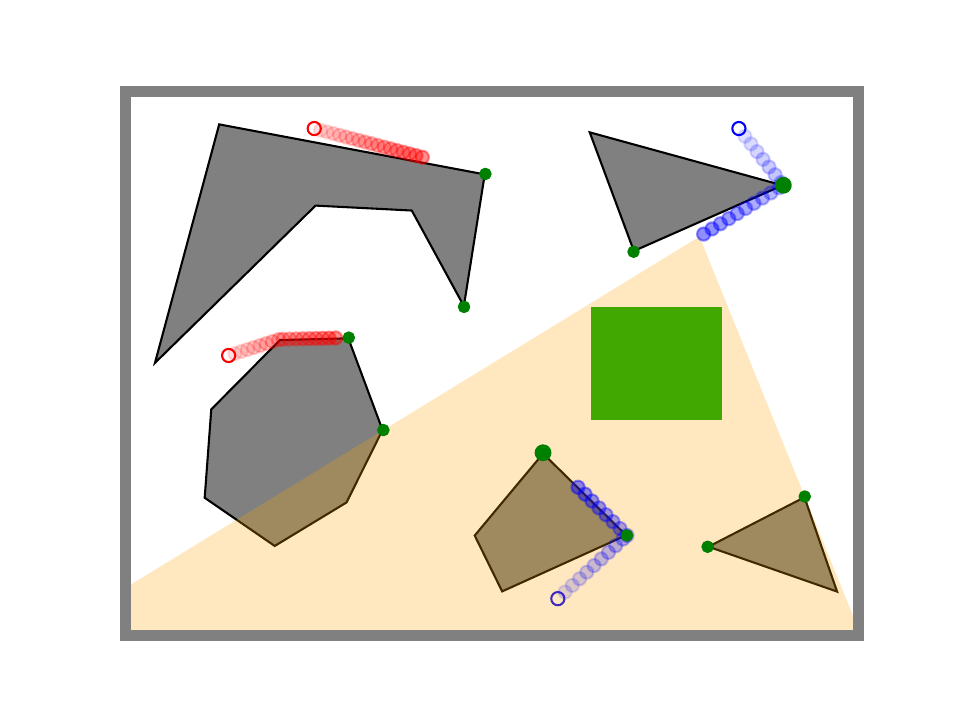}
    \put(-90,-6){\scriptsize$(3b)$} 
    \includegraphics[width=0.32\hsize,height=\height\hsize]{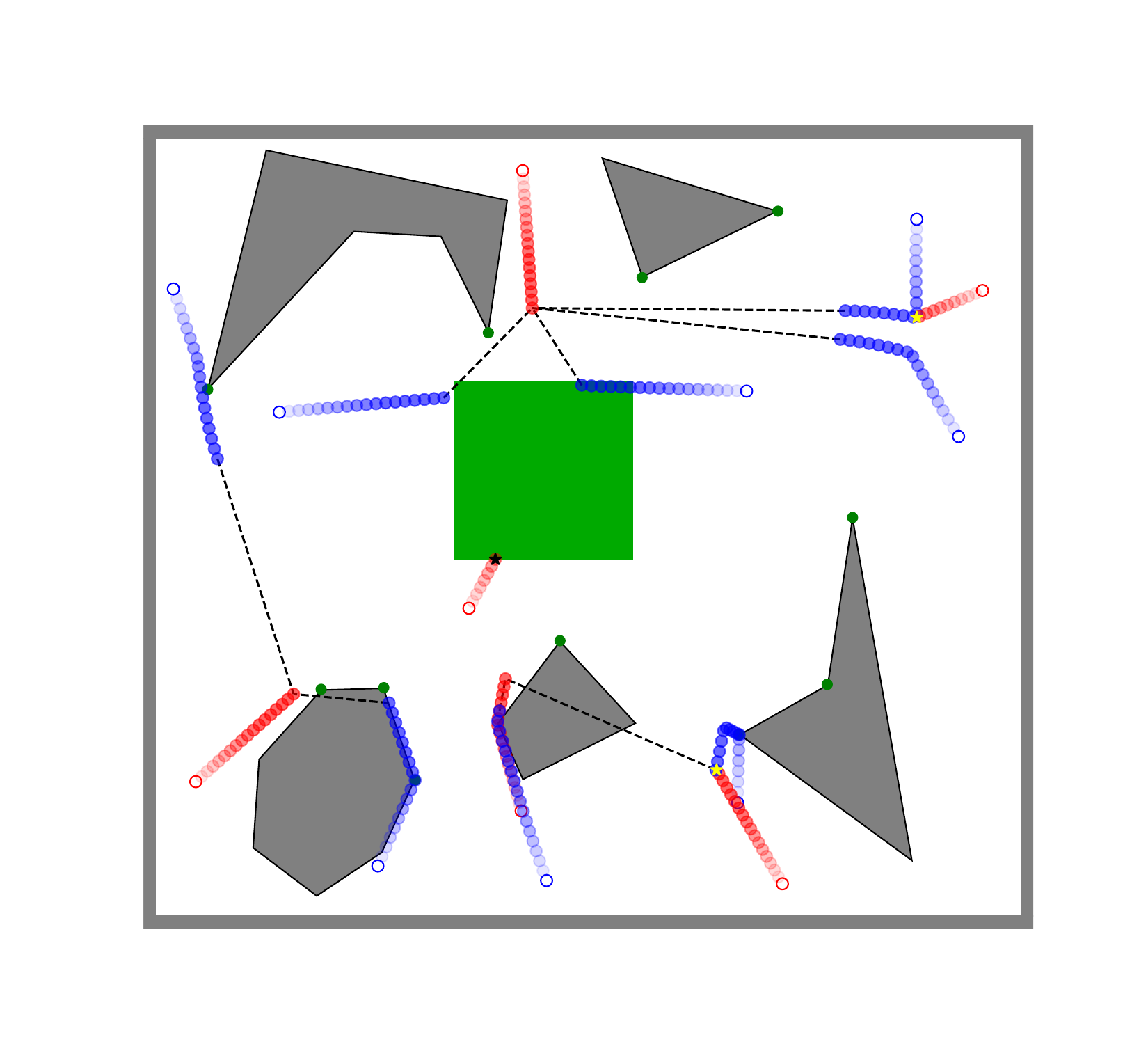}
    \put(-90,-6){\scriptsize$(4b)$} 
    
    \includegraphics[width=0.32\hsize,height=\height\hsize]
    {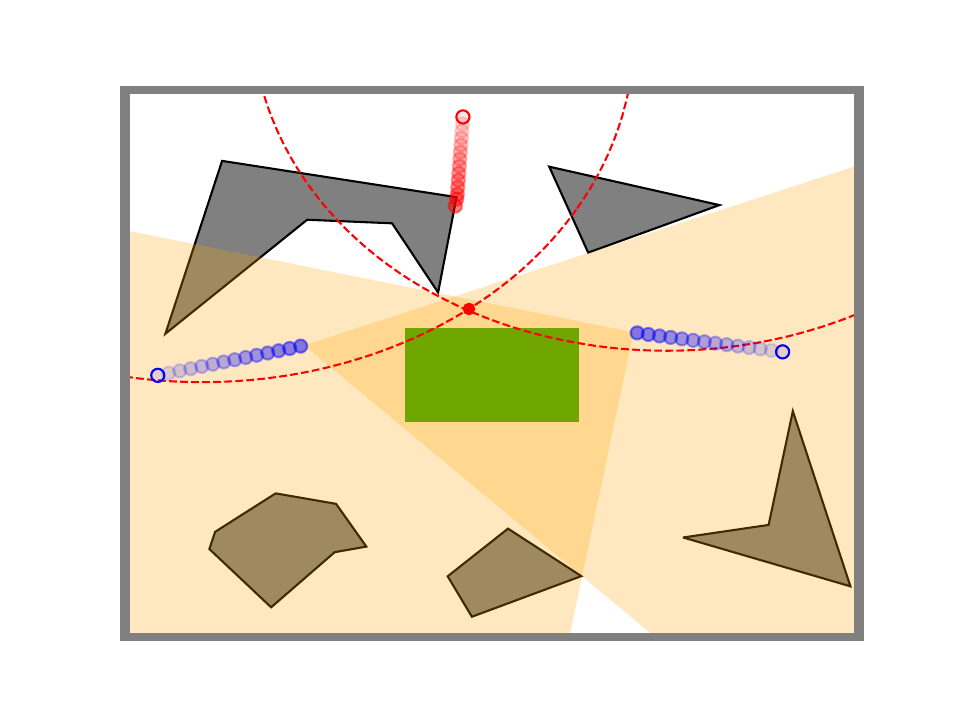}
    \put(-90,-6){\scriptsize$(2b)$} 
    \includegraphics[width=0.32\hsize,height=\height\hsize]{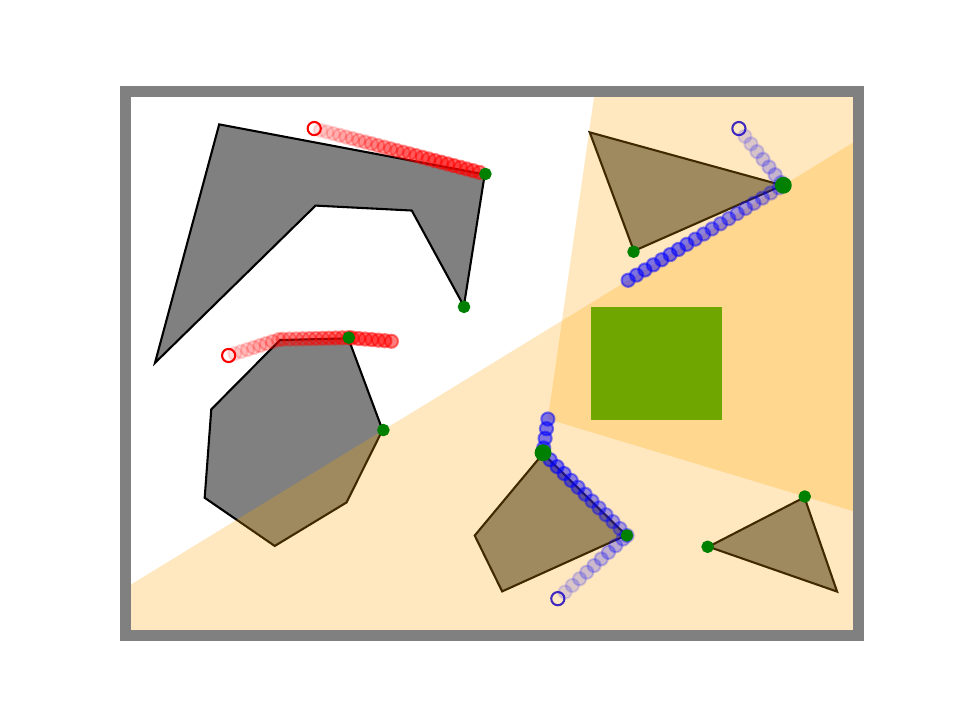}
    \put(-90,-6){\scriptsize$(3c)$} 
    \includegraphics[width=0.32\hsize,height=\height\hsize]{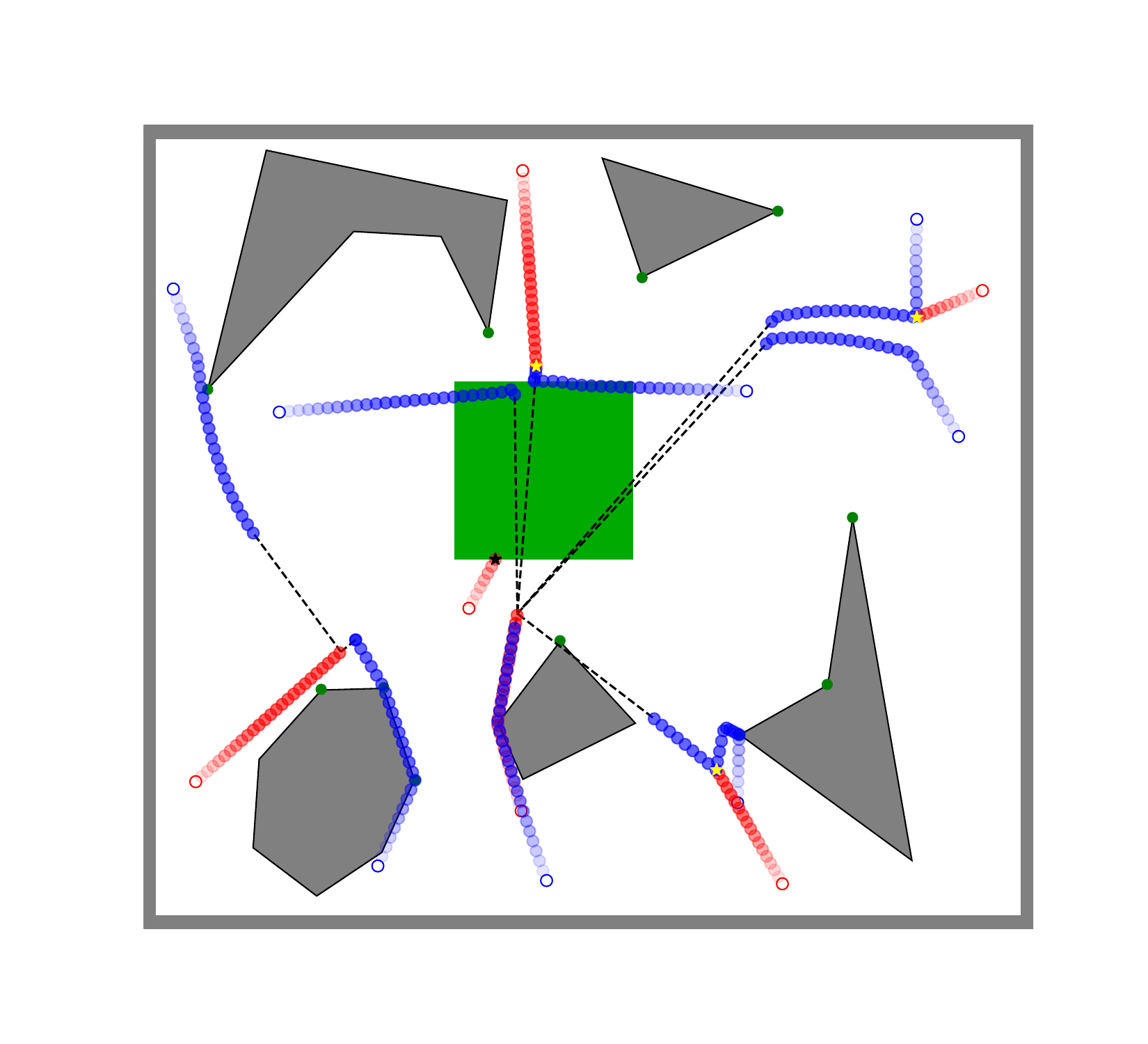}
    \put(-90,-6){\scriptsize$(4c)$}
    
    \includegraphics[width=0.32\hsize,height=\height\hsize]{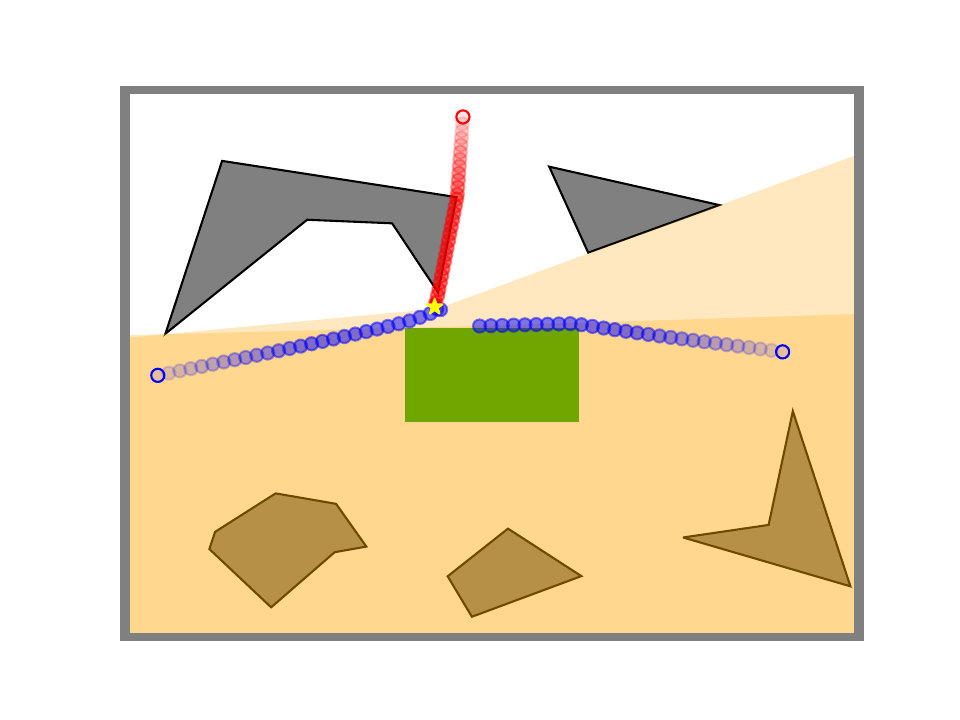}
    \put(-90,-6){\scriptsize$(2c)$} 
    \includegraphics[width=0.32\hsize,height=\height\hsize]
    {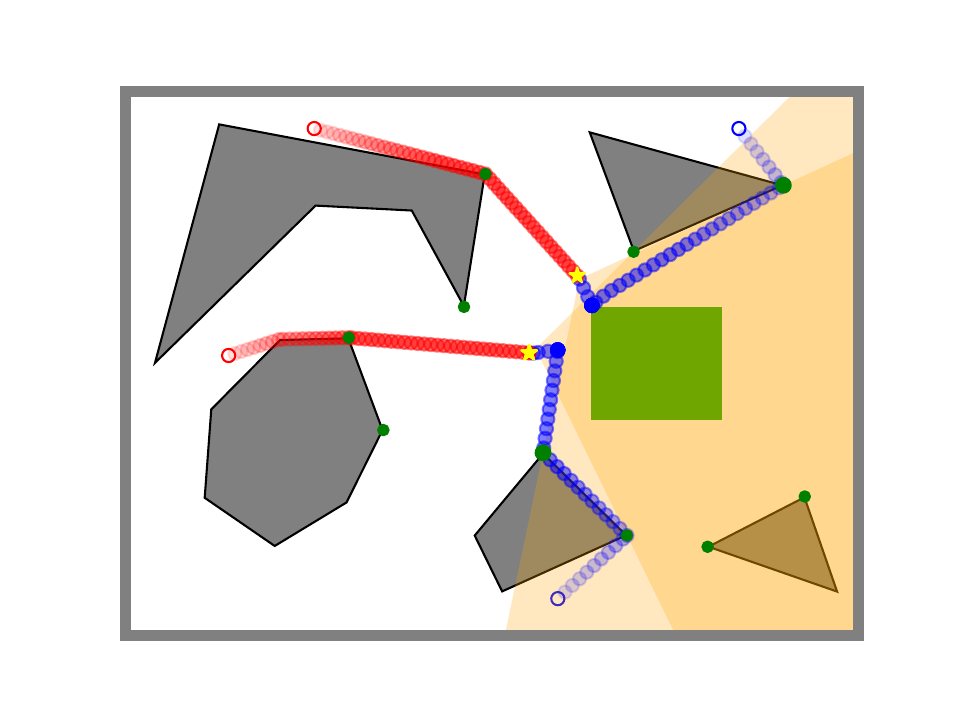}
    \put(-90,-6){\scriptsize$(3d)$} 
    \includegraphics[width=0.32\hsize,height=\height\hsize]{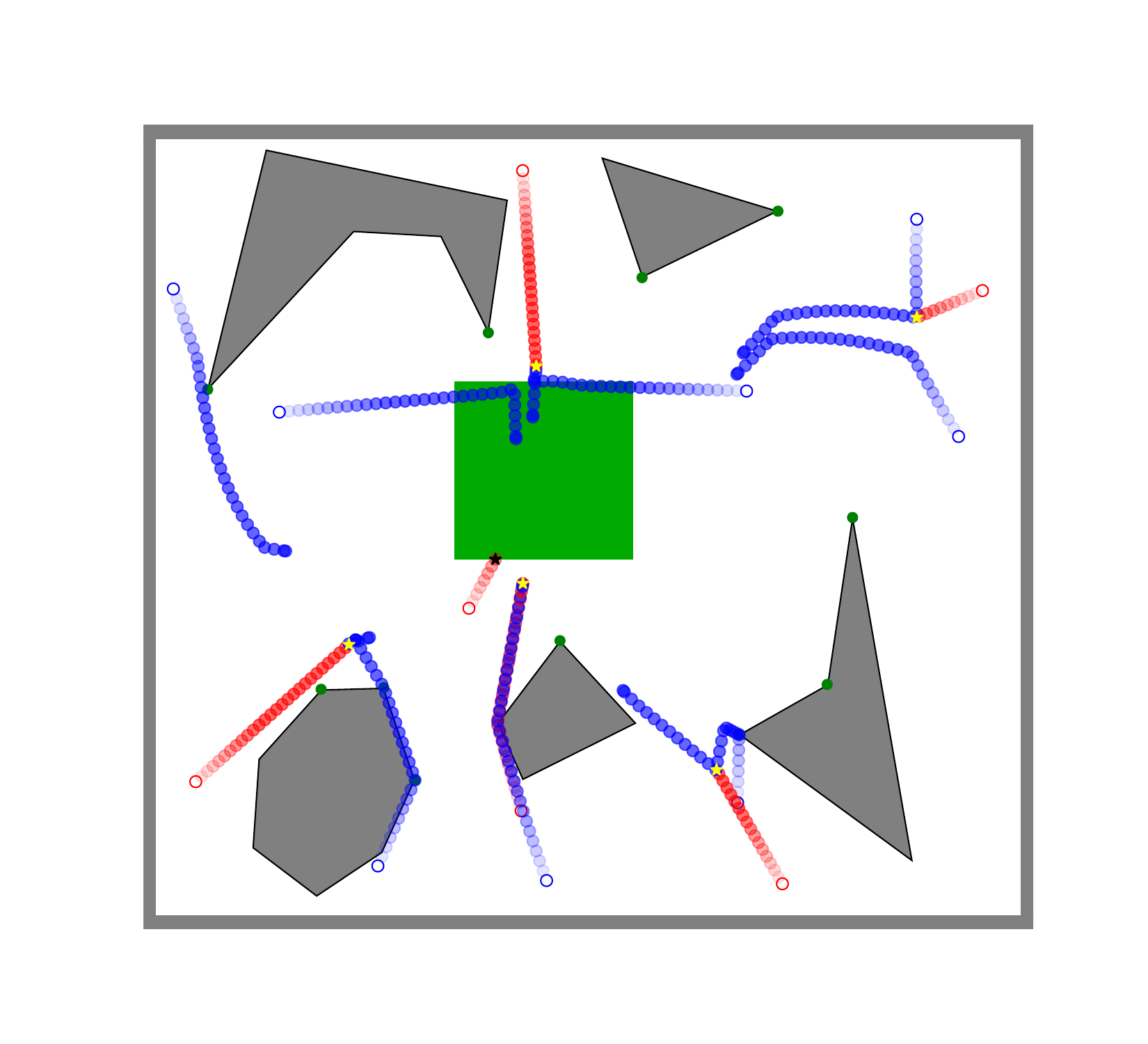}
    \put(-90,-6){\scriptsize$(4d)$} 
    \caption{Four simulations with the MOCG pursuit strategy. $(1a)$ Onsite pursuit winning for three pursuers and one evader. $(2a)$-$(2c)$ Goal-visible pursuit winning for two pursuers and one evader. $(3a)$-$(3d)$ Non-goal-visible pursuit winning for two pursuers and two evaders. $(4a)$-$(4d)$ Three pursuit winnings for eight pursuers and six evaders.}
    \label{fig:simulation}
\end{figure*}

In order to illustrate the MOCG pursuit strategy, we run the multiplayer reach-avoid differential games in various scenarios with different obstacle shapes and distributions, team sizes and initial configurations. We use the first-visible obstacle vertices to construct convex GCPs (see \thomref{thom:construct-convex-GCP}) if necessary. In all scenarios, the strategies of the evaders are generated randomly which are unknown to the pursuers. The games in all scenarios take less than one minute as the pursuit winning regions and strategies have closed forms or require solving simple convex optimization problems.

\startpara{Case 1: onsite pursuit winning} Consider three pursuers $P_1$, $P_2$, $P_3$ and one evader $E_1$ in Fig.~\ref{fig:simulation}(1a). The initial positions of each pursuer and $E_1$ satisfy \eqref{eq:onsite-winning-condition-1v1}, as they are close and there is no obstacle nearby. Thus, by \thomref{thom:onsite-pursuit-winning}, each pursuer can ensure the onsite pursuit winning against $E_1$ individually via the strategy in \lemaref{lema:onsite-pursuit-strategy-1v1}, i.e., capture $E_1$ in the corresponding dotted expanded Apollonius circle in a finite time, regardless of $E_1$'s strategy. Furthermore, \lemaref{lema:onsite-winning-nv1} demonstrates that $E_1$ will be captured in the intersection of three regions bounded by the expanded Apollonius circles. In this case, $P_1$ captures $E_1$ at the yellow star.

\startpara{Case 2: goal-visible pursuit winning} Consider two pursuers $P_1$, $P_2$ and one evader $E_1$ in Fig.~\ref{fig:simulation}(2a). Their initial positions do not satisfy \eqref{eq:onsite-winning-condition-1v1} but conditions in \thomref{thom:goal-visible-winning}. In Fig.~\ref{fig:simulation}(2a), two pursuers are goal-visible with the orange direction ranges. The safe distance is positive as the red closest point in the evasion region to the goal region is outside the goal region. Thus, two pursuers can ensure the goal-visible pursuit winning via the strategy \eqref{eq:goal-visible-pursuit-strategy}, regardless of $E_1$'s strategy. In the snapshot of Fig.~\ref{fig:simulation}(2b), two pursuers are heading towards the closest point as this point is in their direction ranges. The snapshot of Fig.~\ref{fig:simulation}(2c) shows that $P_1$ captures $E_1$ at the yellow star.

\startpara{Case 3: non-goal-visible pursuit winning} Consider two pur- suers $P_1$, $P_2$ and two evaders $E_1$, $E_2$ in Fig.~\ref{fig:simulation}(3a). Their initial positions do not satisfy \eqref{eq:onsite-winning-condition-1v1} and conditions in \thomref{thom:goal-visible-winning}. In Fig.~\ref{fig:simulation}(3a), two pursuers are not goal-visible and all goal-visible obstacle vertices are in green. $P_1$ and $E_1$ (also $P_2$ and $E_2$) meet the conditions in \thomref{thom:non-goal-pursuit-wining} initially. Thus, $P_1$ (resp., $P_2$) ensures the non-goal-visible pursuit winning against $E_1$ (resp., $E_2$) via the strategy \eqref{eq:non-goal-visible-pursuit-strategy}, despite $E_1$'s (resp., $E_2$'s) strategy. The goal-visible obstacle vertices that two pursuers are heading for are marked as larger green vertices. In the snapshot of Fig.~\ref{fig:simulation}(3b), $P_2$ has passed the vertex and become goal-visible with the orange direction range, while $P_1$ has not yet. In the snapshot of Fig.~\ref{fig:simulation}(3c), $P_1$ and $P_2$ are both goal-visible. In the snapshot of Fig.~\ref{fig:simulation}(3d), two pursuers capture the evaders using the onsite pursuit winning strategy as in the final period, they are close with no {obstacles} nearby.

\startpara{Case 4: Three pursuit winnings} Consider eight pursuers $P_i$ ($i=1,\dots,8$) and six evaders $E_j$ ($j=1,\dots,6$) in Fig.~\ref{fig:simulation}(4a). The MOCG pursuit strategy initially generates six matchings indicated by dotted lines: two two-to-one and four one-to-one matching pairs. The capture matchings include: 1) $P_4$ ensures the onsite pursuit winning against $E_5$; 2) $P_6$ and $P_7$ ensure the goal-visible pursuit winning against $E_6$; 3) $P_1$ (resp., $P_3$) ensures the non-goal-visible winning against $E_1$ (resp., $E_4$). The enhanced matching is $P_5$ against $E_5$ via the onsite pursuit winning. There is one closest matching between $P_8$ and $E_3$. Therefore, the MOCG pursuit strategy can guarantee to win against four evaders initially. In the snapshot of Fig.~\ref{fig:simulation}(4b), $E_3$ is able to reach the goal region, while $E_4$ and $E_5$ have been captured respectively by $P_3$ and $P_5$ who continue to capture other evaders. In the snapshot of Fig.~\ref{fig:simulation}(4c), the pursuers still guarantee to defeat four evaders where three have been captured and one (i.e., $E_1$) is being pursued. In the snapshot of Fig.~\ref{fig:simulation}(4d), the pursuers are able to capture five evaders, which demonstrates that the MOCG pursuit strategy can achieve an increasing number of defeated evaders as the game evolves. 

\section{Conclusion} \label{sec:conclusion}

We presented the MOCG pursuit strategy for multiplayer reach-avoid differential games in polygonal environments with general polygonal obstacles. This strategy provides a lower bound on the number of defeated evaders and continually improve the lower bound if a task allocation that can win against more evaders {is found} as the game evolves.  This strategy does not require the state space discretization {as} in many HJ-based approaches and is computationally efficient. This is because all pursuit winning regions and strategies involved either have closed forms or are computed by solving simple convex optimization problems. The three proposed winnings cover three common scenarios. The onsite pursuit winning corresponds to the scenario where the pursuers are close to the evader with no obstacles nearby. The goal-visible case for the close-to-goal pursuit winning corresponds to the scenario where the pursuers can visibly see the whole goal region, while the non-goal-visible case corresponds to the scenario where the visibility fails. The three main employed concepts, expanded Apollonius circles, convex GCPs and ESPs, show that computational geometry methods are powerful in solving games with obstacles. The hierarchical task assignment prioritizes the number of defeated evaders and is complete as it generates the tasks for all pursuers. Future work will involve distributed games, games with limited visible areas and more complicated reconnaissance games with obstacles \cite{YL-EB:23}.

\bibliographystyle{plainurl} 
\bibliography{reference}

\vfill

\end{document}